\documentclass[11pt]{article}
\usepackage[letterpaper, margin=1in]{geometry}
\usepackage{times}
\usepackage[T1]{fontenc}
\usepackage{amsmath, amsthm, amssymb, mathtools}
\usepackage{thm-restate}
\usepackage[colorlinks, citecolor=blue, linkcolor=blue, urlcolor=blue]{hyperref}
\usepackage[capitalise]{cleveref}
\usepackage{microtype}

\hypersetup{
    pdftitle={Faster Deterministic Approximation for Counting Knapsack},
    pdfsubject={Deterministic approximation algorithm for counting knapsack}
}

\theoremstyle{plain}
\newtheorem{theorem}{Theorem}[section]
\newtheorem{lemma}[theorem]{Lemma}
\newtheorem{claim}[theorem]{Claim}
\newtheorem{fact}[theorem]{Fact}
\theoremstyle{definition}
\newtheorem{definition}[theorem]{Definition}

\newcommand{\eps}{\varepsilon}
\renewcommand{\Pr}{\operatorname*{\mathbf{Pr}}}
\newcommand{\poly}{\operatorname{\mathrm{poly}}}
\newcommand{\polylog}{\poly\log}
\newcommand{\cO}{\mathcal{O}}
\newcommand{\tcO}{\widetilde{\mathcal{O}}}
\newcommand{\supp}{\mathrm{supp}}
\newcommand{\numberthis}{\addtocounter{equation}{1}\tag{\theequation}}

\title{Faster Deterministic Approximation for Counting Knapsack}
\author{
    Ce Jin\thanks{University of California Berkeley. \textbf{\texttt{cejin@berkeley.edu}}.
        Supported by the Miller Research Fellowship at the Miller Institute for Basic Research
        in Science, UC Berkeley.}
    \and
    Jaewoo Park\thanks{University of California San Diego. \textbf{\texttt{pjaewoo@ucsd.edu}}.
        Supported by NSF HDR TRIPODS Phase II grant 2217058 (EnCORE Institute).}
    \and
    Barna Saha\thanks{University of California San Diego. \textbf{\texttt{bsaha@ucsd.edu}}.
        Supported by NSF HDR TRIPODS Phase II grant 2217058 (EnCORE Institute).}
}
\date{\vspace{-5ex}}

\begin{document}
\maketitle

\begin{abstract}
    Given $n$ nonnegative integer weights $w_1, \ldots, w_n$ and an integer knapsack capacity $T$,
the $\#$Knapsack problem asks to count the number of subsets $I \subseteq \{1, 2, \dots, n\}$
such that $\sum_{i \in I} w_i \le T$.
A long line of research has been devoted to designing faster
$(1 \pm \varepsilon)$-approximation algorithms for this \#P-hard problem.
Until very recently, the best-known time complexity was
$\widetilde{\mathcal{O}}\bigl(n^{2.5}/\mathrm{poly}(\varepsilon)\bigr)$, achieved by the
randomized algorithm of [Dyer, STOC 2003] and the deterministic algorithm of [Gawrychowski,
Markin, and Weimann, ICALP 2018].
The recent work by [Feng and Jin, SODA 2025] improved the randomized complexity to
$\widetilde{\mathcal{O}}\bigl(n^{1.5}/\mathrm{poly}(\varepsilon)\bigr)$, leaving open the
question whether this $n^{2.5}$ barrier could also be broken by a \emph{deterministic}
algorithm.

In this work, we answer this question in the affirmative by presenting a deterministic
$(1 + \varepsilon)$-approximation algorithm for \#Knapsack in
$\widetilde{\mathcal{O}}\bigl(n^{2.4}/\mathrm{poly}(\varepsilon)\bigr)$ time.
Our key technical innovation is to deterministically simulate the balancing technique,
which was previously used in the literature of knapsack-type problems via random
partitions and seemed to require randomization.
\end{abstract}

\newpage
\section{Introduction}
\label{sec:introduction}
We consider one of the most basic counting problems, namely approximately counting the number of
solutions of the $0/1$ knapsack problem (\#Knapsack). The \#Knapsack problem is defined as follows:
given $n$ nonnegative integer weights $w_1, \ldots, w_n$ and an integer knapsack capacity $T$,
the goal is to count the number of subsets whose total weight is at most $T$; that is, to compute
\[
    \Bigl\lvert \Bigl\{\, I \subseteq \{1,2,\ldots,n\} : \sum_{i \in I} w_i \le T \,\Bigr\} \Bigr\rvert.
\]
The \#Knapsack problem is equivalent to the task of counting the vertices of the \(n\)-dimensional
Boolean hypercube that lie in a given halfspace. This geometric perspective reveals its connections
to multiple other areas in theoretical computer science, including learning
theory~\cite{KKMS08, OS11, DDFS14}, property testing~\cite{MORS10},
pseudorandomness~\cite{DGJSV10, RS10}, and circuit complexity~\cite{PS90, HG91, Sherstov09, KW16}.

A related decision problem, Subset Sum\footnote{ The Subset Sum problem asks whether there exists
$I\subseteq[n]$ such that $\sum_{i\in I}w_i=T$.} is well known to be NP-hard, but admits a
pseudo-polynomial $\poly(nT)$-time algorithm based on standard dynamic programming.
The counting problem \#Knapsack is known to be \#P-complete. However, the same dynamic programming 
approach can be adapted to solve it exactly in pseudo-polynomial $\poly(nT)$ time.

Given the fundamental nature of this problem and the hardness of \#Knapsack, a long line of work
has focused on designing
$(1\pm\eps)$-approximation algorithms for \#Knapsack with running time polynomial in $n$ and $1/\eps$,
where $\eps>0$ is the
error parameter. These works include both randomized algorithms
(fully polynomial randomized approximation schemes, FPRASes) and deterministic algorithms
(fully polynomial-time approximation schemes, FPTASes).

Early \emph{randomized} approximation algorithms for \#Knapsack used Markov chain Monte Carlo (MCMC) sampling
together with the sampling-to-counting reduction of Jerrum et al.~\cite{JERRUM86}.
Dyer et al.~\cite{DFKKPV93} analyzed a random walk for sampling feasible subsets almost uniformly,
yielding a $2^{\cO(\sqrt{n}(\log n)^{2.5})}/\eps^2$-time approximation algorithm for \#Knapsack.
Morris and Sinclair~\cite{MS99} later improved the mixing
time of this random walk, yielding the first FPRAS for
\#Knapsack with running time
$\cO(n^{4.5+\eta}+n^2/\eps^2)$,
for any fixed constant $\eta>0$.
In 2003, Dyer~\cite{Dyer03} improved the FPRAS running
time to $\tcO(n^{2.5}+n^2/\eps^2)$\footnote{$\tcO$ hides $\polylog(n/\eps)$ factors.} using a completely
different method. For constant $\eps$, this remained
the best-known running time for over twenty years. Instead
of using Markov chains, he used textbook dynamic programming on discretized capacities via
randomized weight rounding, as well as the sampling-to-counting reduction. Very recently, Feng and
Jin~\cite{FJ25} improved Dyer's
approach to obtain a running time of $\tcO(n^{1.5}\eps^{-2})$, achieving the first subquadratic dependence on $n$.

Deterministic approximation schemes for \#P-complete problems are relatively rare, counting knapsack
solutions being one of those rare instances. Gopalan et al.~\cite{GKM10} and Štefankovič et
al.~\cite{SVV12} (see also their combined work~\cite{GKRSVV11}) were the first to give an FPTAS for
\#Knapsack, with a running time of $\cO(n^{3}\eps^{-1} \log(n/\eps))$. Their methods are based on
dynamic programming: Gopalan et al.\ count the number of accepting paths of the small-width
read-once branching programs they construct, while \v{S}tefankovi\v{c} et al.\ approximate a
``dual'' dynamic program that computes the minimum capacity required to obtain at least a given number
of feasible subsets. Rizzi and Tomescu~\cite{RT14} later improved the running time by a logarithmic factor 
using floating-point arithmetic to approximate the counting values in the dynamic programming table.
Halman~\cite{Halman16} subsequently gave a simpler FPTAS for the more general integer knapsack
counting problem, with running time $\tcO(n^3/\eps)$ for the 0/1 case. Finally, Gawrychowski, Markin, and
Weimann~\cite{GMW18} obtained the previously fastest deterministic FPTAS for \#Knapsack, with running time
$\tcO(n^{2.5}\eps^{-1.5})$, using divide-and-conquer and fast convolution of sparsified counting functions.

These developments raise the following question:
\begin{center}
    {\bf Question:} {\it Does there exist an FPTAS for \#Knapsack with a running time of
    $\cO(n^{2.5 - \delta}/\poly(\eps))$ for some constant $\delta > 0$?}
\end{center}

\paragraph*{Our Results.}
In this paper, we present an improved deterministic FPTAS
for \#Knapsack that breaks the $n^{2.5}$ barrier,
answering the above question affirmatively. Compared to the previous state-of-the-art, 
which runs in $\tcO(n^{2.5}\eps^{-1.5})$ time, our algorithm improves
both the $n$-dependence and the $\eps$-dependence by
polynomial factors.
\begin{restatable}[]{theorem}{fptas}
    \label{thm:count-knapsack-fptas}
    There exists an FPTAS for \#Knapsack that runs in $\tcO(n^{2.4}/\eps^{1.4})$ time.
\end{restatable}

We obtain our new result by using the \emph{balancing technique} to improve the previous framework
of Gawrychowski, Markin, and Weimann \cite{GMW18}. The balancing technique has been widely used in
the literature of knapsack-type problems. However, previous applications of this technique were
inherently randomized, as they crucially rely on random partitions and concentration inequalities,
e.g., \cite{ BateniHSS18,BringmannC23,HeX24,BringmannDPesa24,ChenLMZfocs24}. Our key technical
innovation is to deterministically simulate the random partition in the context of approximating
\#Knapsack. We give a more detailed overview below.

\subsection{Technical Overview}

\paragraph*{Prior works on FPTAS for \#Knapsack.}
It is well known that \#Knapsack with input integer
weights $w_1,\ldots,w_n$ and capacity $T$ can be solved
exactly in pseudopolynomial $\poly(nT)$ time by a
standard dynamic program that maintains the
\emph{count function}
\[
    f_G(t):=
    \bigl|\{I\subseteq G:\sum_{i\in I}w_i=t\}\bigr|
\]
for $G=[i]$\footnote{Throughout, denote $[i] = \{1, 2, \dots, i\}$.}, where $0\le i\le n$, represented as a
list indexed by $0\le t\le T$.

To avoid the prohibitive dependence on $T$ of this dynamic program, previous FPTASes for \#Knapsack
\cite{GKRSVV11,RT14,Halman16,GMW18} used approximation techniques to compress the count functions 
while allowing a $(1+\eps)$ approximation factor. Here, the right notion is
\emph{$(1 + \eps)$-sum-approximation} (using the terminology of \cite{GMW18}; see
\cref{def:sum-approximation} for details), namely that the prefix sum
$f_G^{\le}(t) \coloneqq \sum_{t' \le t} f_G(t')$ of $f_G$ should be within a factor of $(1 + \eps)$
of the prefix sum $\widetilde f_G^{\le}(t)$ of the sparsified function $\widetilde f_G$ for all $t$.
In view of the simple bound $f_G^{\le}(t) \le 2^{|G|}$, one can find such a
$(1 + \eps)$-sum-approximation $\widetilde f_G$ which only has
$\log_{1 + \eps}2^{|G|} = \tcO(|G|/\eps)$ non-zero entries, remarkably independent of the knapsack
capacity $T$.

Using sum-approximations, Gawrychowski, Markin, and
Weimann~\cite{GMW18} gave a simplified version of
Halman's algorithm~\cite{Halman16} that runs in
$\cO(n^3/\eps)$ time. They further improved the dependence on $n$ to
$\tcO(n^{2.5}/\eps^{1.5})$. At a high level, they replaced
the simple dynamic programming (which adds one item into consideration each time) by a
divide-and-conquer approach. For disjoint subsets of input items $G_1, G_2$, the count function of
$G_1 \cup G_2$ is the \emph{convolution} of the count functions of $G_1$ and $G_2$, i.e.,
$f_{G_1 \cup G_2}(t) = \sum_{t_1 + t_2 = t} f_{G_1}(t_1) \cdot f_{G_2}(t_2)$. Hence, one can evenly
partition the input items into $G_1 \cup G_2$, recursively compute their (sparsified) count
functions, and then return their (sparsified) convolution. This is sometimes called a
``partition-and-convolve'' paradigm, which is common in the literature of Knapsack-type problems
\cite{KX19,Bring17,BateniHSS18,DJM23, BringmannDPesa24,FJ25} with different notions of
``convolution'' tailored to different settings.

\paragraph*{Balancing via random partitioning.}
Our improvement is inspired by the following common idea that has been frequently used in recent
\emph{randomized} algorithms for Knapsack-type problems. Fix any (unknown) feasible knapsack subset
$I \subseteq G$ whose total weight is $t$ (here and below, one may focus on the top level of
recursion and think of $G = [n]$, but the same discussion also applies to lower levels of recursion,
where $G$ is a subset of the input items $[n]$). If we randomly partition $G$ into $G_1 \cup G_2$,
then with high probability $I$ is balanced with respect to this partition, i.e., the total weights
of $I \cap G_1$ and $I \cap G_2$ are both concentrated in an interval of radius
$\tcO (\sqrt{|G|} w)$ around their common expectation $t/2$, where $w$ denotes the largest individual item weight
in $G$. Therefore, when convolving the answers from $G_1$ and from $G_2$, it suffices to only range
over $t_1 + t_2 = t$ such that $|t_1 - t_2| \le \tcO (\sqrt{|G|} w)$ while still capturing the
subset $I$, thereby reducing the time complexity of convolution computation. This balancing idea, in
various forms, often in conjunction with partition-and-convolve, has found applications in
pseudopolynomial-time algorithms for Subset Sum \cite{ChenLMZfocs24} and 0-1 Knapsack
\cite{BateniHSS18,BringmannC23,HeX24,BringmannDPesa24}.\footnote{Many other works on Knapsack-type
problems also used random partitioning, but only required a \emph{weaker} concentration guarantee,
namely that when $G$ is randomly partitioned into $g$ groups, each group has sum $\tcO(|G|w/g)$ with
high probability. Examples include the classic randomized Subset Sum algorithm of Bringmann
\cite{Bring17}, as well as the recent randomized FPRAS for \#Knapsack by Feng and Jin \cite{FJ25}.
In this context, a recent independent work by Chan \cite{Chan26} successfully derandomized
Bringmann's Subset Sum algorithm. However, Chan's techniques seem insufficient for derandomizing
Feng and Jin \cite{FJ25}'s algorithm, since a more inherent source of randomness in their algorithm
is a Monte Carlo estimation step originating from Dyer \cite{Dyer03}'s FPRAS. }

Unsurprisingly, all the aforementioned algorithms are inherently randomized. It is not clear how one
would derandomize this approach: in order to ensure every possible $I \subseteq G$ is balanced, a
single deterministic partition $G = G_1 \cup G_2$ clearly does not suffice. A natural attempt would
be to try multiple partitions from a pseudo-random family, but this would hugely increase the time
complexity. Nevertheless, we manage to derandomize this balancing approach using just a single
deterministic partitioning.

\paragraph*{Deterministic balancing.}
In principle, such derandomization is possible because it is \emph{not} necessary for the partition
$G = G_1 \cup G_2$ to balance \emph{every single} possible feasible subset $I \subseteq G$. Instead,
we only aim to balance \emph{a $(1 - \eps)$-fraction} of the feasible subsets $I$, which is
tolerable since we are allowed to incur $\eps$ relative error in the count.

We provide intuition on how such deterministic partitioning can be constructed. For simplicity,
we first illustrate our ideas under the \emph{bounded-ratio} assumption that $w_i \in [w/2, w]$ for
all items $i \in G$. Later, we will describe our additional technical ideas for removing this
assumption. To simplify notation and illustration, we also assume the indices of the items are
$G = \{1, 2, \dots, |G|\}$, and $|G|$ is an even number.

For $I \subseteq G$, let
$w(I) \coloneqq \sum_{i \in I} w_i$. We show how to
deterministically construct a partition $G = G_1 \cup G_2$, such that for every capacity $T$, among
all feasible subsets
$\Omega \coloneqq \{I \subseteq G: w(I) \le T\}$, at least a
$(1 - \eps)$-fraction of the subsets $I \in \Omega$ are balanced in the sense that
\begin{equation}
    \label{eq:intro-weight-discrepancy}
    \Bigl\lvert w(I \cap G_1) - w(I \cap G_2)
    \Bigr\rvert \le \tcO\bigl(\sqrt{|G|\log(1/\eps)} \cdot w\bigr).
\end{equation}
Note that we still need to argue why \cref{eq:intro-weight-discrepancy} helps speed up the convolution computation in
the framework of \cite{GMW18}. This is intuitive but not obvious and we shall come back to this
shortly.

We start with the toy case where all items $i \in G$ have equal weight $w_i = w$. In this case, we
observe that any partition $G=G_1\cup G_2$ with $|G_1|=|G_2|=|G|/2$ already
satisfies \cref{eq:intro-weight-discrepancy} for at least a $(1 - \eps)$-fraction of the subsets
$I \in \Omega$. Indeed, since all items are symmetric, $\Omega = \binom{G}{\le \lfloor T/w\rfloor}$,
so a uniformly random $I \in \Omega$ satisfies \cref{eq:intro-weight-discrepancy} with probability
at least $(1 - \eps)$ by a standard Chernoff bound.

The above symmetry-based argument can be vastly generalized to the case where $G$ consists of pairs
of equal-weight items, i.e., $w_1 = w_2, w_3 = w_4, \dots, w_{|G| - 1} = w_{|G|}$. In this case, we
claim that the desired balancing property can be satisfied by choosing
$G_1 = \{1, 3, 5, \dots\}, G_2 = \{2, 4, 6, \dots\}$ which perfectly separates every pair. To see
this, observe that a uniformly random $I \in \Omega$ can be sampled in two stages:
\begin{itemize}
    \item First, jointly determine the number of selected items from each pair $\{2k - 1, 2k\}$ according to their distribution
        under a uniformly random $I\in\Omega$. Each pair has zero, one, or two selected items.
    \item Then, if exactly one of $\{2k - 1, 2k\}$ should be included in $I$, flip a coin to decide
        which one, so that it contributes to $G_1 \cap I$ or $G_2 \cap I$ each with $1/2$ probability.
        The coin flips are independent for all $k$.
\end{itemize}
By the symmetry between even- and odd-indexed items, one can apply a Chernoff bound to the
randomness in the second stage, and show that \cref{eq:intro-weight-discrepancy} indeed holds with at least
$(1 - \eps)$ probability over a random $I \in \Omega$.

Now we consider the general bounded-ratio case with
$w/2 \le w_1 \le w_2 \le \dots \le w_{|G|} \le w$. It turns out that our even-odd indexed
partitioning $G = G_1 \cup G_2$ still satisfies the desired balancing property with only a small loss,
worsening $\log(1/\eps)$ in \cref{eq:intro-weight-discrepancy} to $\log(|G|/\eps)$. The idea is to relate
the weight sequence $(w_1, w_2, \dots, w_{|G|})$ with its rounded version,
$(\widetilde w_1, \widetilde w_2, \dots, \widetilde w_{|G|}) \coloneqq (w_2, w_2, w_4, w_4, \dots)$.
The total additive error between them is only
\begin{equation}
    \label{eq:intro-pair-rounding-error}
    (w_2 - w_1) + (w_4 - w_3) + \dots + (w_{|G|} - w_{|G| - 1}) \le w_{|G|} - w_1 \le w/2.
\end{equation}
By the previous paragraph, we know the desired balancing property already holds for this rounded
version of the weight sequence. Intuitively, this property should carry over to the original version
as well, with small loss. Indeed, in the next paragraph we provide a proof that justifies this
intuition, by adapting Dyer's \cite{Dyer03} mapping argument which was originally introduced in his
randomized FPRAS for \#Knapsack.

Let $\widetilde w(I) \coloneqq \sum_{i \in I} \widetilde w_i$ for any subset $I \subseteq G$. Consider
the collection of feasible rounded subsets
$\widetilde \Omega \coloneqq \{I \subseteq G: \widetilde w(I) \le T\} \subseteq \Omega$.
Define a mapping $\pi$ from $I \in \Omega \setminus \widetilde \Omega$ by letting
$\pi(I) = I \setminus \{i\}$ for an arbitrary item $i \in I$. Observe that
$\pi(I) \in \widetilde \Omega$ must hold, since the weight of the removed item, $w_i \ge w/2$,
offsets the weight difference $\widetilde w(I) - w(I)$ incurred when
rounding up (by \cref{eq:intro-pair-rounding-error}). Additionally define $\pi(I) = I$ for
$I \in \widetilde \Omega$, so that $\pi$ maps from $\Omega$ to $\widetilde \Omega$. Note the map
$\pi$ does not affect the balancedness very much: If $I \in \Omega$ has large discrepancy,
\begin{equation}
    \label{eq:intro-original-large-discrepancy}
    \Bigl\lvert w(I \cap G_1)  - w(I \cap G_2)
    \Bigr\rvert \ge c\sqrt{|G|\log (|G|/\eps)} w + 2w,
\end{equation}
then $Y = \pi(I) \in \widetilde \Omega$ also has large discrepancy,
\begin{equation}
    \label{eq:intro-rounded-large-discrepancy}
    \Bigl\lvert \widetilde w(Y \cap G_1) - \widetilde w(Y \cap G_2)
    \Bigr\rvert \ge c\sqrt{|G|\log(|G|/\eps)} w.
\end{equation}
Indeed, deleting one item changes the discrepancy by at most $w$, and rounding the remaining items
changes it by at most $w/2$. Applying the concentration bound from the previous
paragraph with failure probability
$\eps/(|G|+1)$, we obtain that, for a sufficiently
large constant $c>0$, there are at most
$\eps|\widetilde\Omega|/(|G|+1)$ sets
$Y\in\widetilde\Omega$ satisfying
\cref{eq:intro-rounded-large-discrepancy}. Now, each such $Y$ has at most $(|G| + 1)$
preimages under $\pi$: $Y$ itself and the possible sets obtained by adding back one removed item.
Therefore, the number of $I \in \Omega$ satisfying \cref{eq:intro-original-large-discrepancy} is at most
$(|G|+1)\cdot \eps|\widetilde\Omega|/(|G|+1)=\eps|\widetilde\Omega|\le\eps|\Omega|$
as desired.

\paragraph*{Applying balancing to sparsified count functions.}
Now we discuss how the balancing property \cref{eq:intro-weight-discrepancy} helps reduce the time complexity
for convolution in \cite{GMW18}'s framework. We are still focusing on the bounded-ratio case.
As discussed earlier, the bound in \cref{eq:intro-weight-discrepancy} means that when
convolving the (sparsified) count functions $\widetilde f_{G_1}$ and $\widetilde f_{G_2}$, we only
need to iterate over the pairs $(t_1, t_2)$ (where $t_i$ stands for the total weight of items picked
from $G_i$) such that $|t_1 - t_2| \le \tcO(\sqrt{|G|\log(1/\eps)} \cdot w)$, and
$\widetilde f_{G_1}(t_1) \neq 0, \widetilde f_{G_2}(t_2) \neq 0$. We call such pairs $(t_1, t_2)$
the useful pairs. The number of useful pairs not only depends on the upper bound on $|t_1 - t_2|$, but also on the
number of non-zero entries of $\widetilde f_{G_1}, \widetilde f_{G_2}$ in the relevant regions,
which in turn depends on how fast the prefix sums $f^\le_{G_1}(t_1), f^\le_{G_2}(t_2)$ grow as
$t_1, t_2$ change. More specifically, if $\widetilde f_{G_2}$ is an $\eps$-sparsified
$(1 + \eps)$-sum-approximation of $f_{G_2}$, then the number of non-zero entries
$\widetilde f_{G_2}(t_2)$ in an interval $t_2 \in [t_L, t_R]$ is at most
\[
    \cO\left(
        1 + \log_{1 + \eps}
        \left(\frac{f^\le_{G_2}(t_R)}{f^\le_{G_2}(t_L)}\right)
    \right).
\]
In our scenario, for every fixed $t_1$, the relevant interval for $t_2$ is
\[
    [t_L, t_R] = [\max\{0, t_1 - r\}, t_1 + r],
    \qquad
    r = \tcO\bigl(\sqrt{|G|\log(1/\eps)}\,w\bigr).
\]

We claim that the ratio $f^\le_{G_2}(t_R)/f^\le_{G_2}(t_L)$ is at most
$n^{\cO(1 + (t_R - t_L)/w)}$, which implies that there are only $\tcO(\sqrt{|G|}/\eps)$ many $t_2$
such that $(t_1, t_2)$ is useful (compared to the naive bound $\tcO(|G|/\eps)$). We use the mapping
argument of Dyer~\cite{Dyer03} again. For every $I \subseteq G_2$ with $w(I) \le t_R$,
remove $\lceil 2(t_R - t_L)/w\rceil$ items from $I$, or all items if fewer are present. The
remaining total weight is at most $t_L$, since every item has weight at least $w/2$ and $t_L \ge 0$.
Fixing a rule for choosing the removed items gives a map from $\Omega_{G_2}[t_R] := \{ I \subseteq G_2 : w(I) \le t_R\}$ to
$\Omega_{G_2}[t_L] := \{ I \subseteq G_2 : w(I) \le t_L\}$, under which every element in the latter set has at most
$n^{\cO(1 + (t_R - t_L)/w)}$ preimages, corresponding to the possible choices of removed items. This
proves the desired ratio bound.

\paragraph*{Relaxing the bounded-ratio assumption.}
A natural approach to handling arbitrary weights is to divide
the items into $\cO(\log T)$ bounded-ratio classes.
However, this introduces an undesirable dependence on $T$.

The bounded-ratio assumption was used both to prove the
balancing property (\cref{eq:intro-weight-discrepancy})
and to bound the number of useful pairs.
In particular, the lower bound $w/2$ on individual weights
allowed us to convert a bound on weight discrepancy into
a bound on the number of items removed in Dyer's mapping
argument. This conversion is no longer possible for
arbitrary weights.

We therefore redefine the balancing property in terms of
the number of items that must be removed, rather than the
weight discrepancy. More precisely, we call a subset
$I\subseteq G$ \emph{$C$-balanced} if removing at most $C$
of the heaviest items from either side (considered
separately) makes its remaining weight at most that of
the other side (see \Cref{def:balanced-subsets}).
Our balancing lemma shows that, with an appropriate
choice of $C=\tcO(\sqrt{|G|})$, the same even-odd
partition balances all but an $\eps$ fraction of
the feasible subsets.

To prove this, we first modify the rounding scheme.
As before, we round each pair of weights so that
$\widetilde w_{2k-1}=\widetilde w_{2k}=w_{2k}$,
but leave the heaviest selected item unrounded.
Specifically, for nonempty $I$, define
\[
    \widetilde w(I)
    :=w_\ell+
    \sum_{i\in I\setminus\{\ell\}}\widetilde w_i,
    \qquad \ell=\max I,
\]
and set $\widetilde w(\varnothing)=0$.
By the same telescoping argument as before,
\[
    0\le\widetilde w(I)-w(I)
    \le\sum_{i<\ell}(\widetilde w_i-w_i)
    \le w_\ell.
\]
Consequently, deleting the heaviest selected item when
necessary maps every original feasible subset to a
rounded feasible subset, with at most $|G|+1$ preimages
per image. Reinserting the deleted item increases the
number of removals needed for balance by at most one.

We analyze a uniformly random rounded feasible subset
using the same independent swaps within adjacent pairs
as before, except that we leave the heaviest selected
item and its pair unchanged. These swaps preserve both
the rounded weight and the uniform distribution.

The main new idea is to separate the selected items into
heavy and light parts. Fix a rounded feasible subset $J'$
before the swaps. If $|J'|\le C$, it is already
$C$-balanced. Otherwise, choose an odd index $h$ so that
the suffix $[h,|G|]$ contains the $C-1$ heaviest selected
items and possibly one additional selected item.
For the subset $J$ after the swaps, write
\[
    H:=J\cap[h,|G|],
    \qquad
    L:=J\cap[1,h-1].
\]
Then $|H|\le C$, every item in $L$ weighs at most $w_h$,
and every item in $H$ weighs at least $w_h$.

Hoeffding's inequality bounds the weight discrepancy
between the two light parts at the scale $w_h$;
the bias caused by unequal weights within pairs is
only $\cO(w_h)$ by telescoping.
Meanwhile, Chernoff's bound guarantees that both sides
receive sufficiently many heavy items.
Choosing the constant and logarithmic factors in
$C=\tcO(\sqrt{|G|})$ sufficiently large therefore
ensures, with high probability, that the heavy items
on either side outweigh the discrepancy between
the two light parts.
In particular,
\[
    w(L\cap G_1)-w(L\cap G_2)
    \le w(H\cap G_2),
\]
and the analogous inequality holds with $G_1$ and
$G_2$ interchanged.
Hence, deleting $H\cap G_1$ from the first side gives
\[
    w\bigl((J\cap G_1)\setminus H\bigr)
    =w(L\cap G_1)
    \le w(L\cap G_2)+w(H\cap G_2)
    =w(J\cap G_2).
\]
Since $|H|\le C$, only $C$ items need to be removed.
Together with the deletion mapping above,
this establishes the balancing guarantee for the
original feasible subsets.

Finally, this deletion-based balancing property also
allows us to bound the number of useful pairs without
assuming bounded ratios.
Since a fixed subset has at most
$|G|^{\cO(C)}$ possible preimages under deletion
of at most $C$ items, we define useful pairs by
\[
    \widetilde f_{G_1}^{\le}(t_1)
    \le 6|G|^C\widetilde f_{G_1}^{\le}(t_2),
    \qquad
    \widetilde f_{G_2}^{\le}(t_2)
    \le 6|G|^C\widetilde f_{G_2}^{\le}(t_1),
\]
rather than by a bound on $|t_1-t_2|$
(see \Cref{def:useful-pairs}).
For each fixed $t_1$, the number of qualifying
$t_2\ge t_1$ is at most
\[
    1+\log_{1+\eps}(6|G|^C)
    =\tcO(C/\eps)
    =\tcO(\sqrt{|G|}/\eps).
\]
The case $t_1\ge t_2$ is symmetric.
Since each support has size $\tcO(|G|/\eps)$,
there are only $\tcO(|G|^{3/2}/\eps^2)$ useful pairs.
The formal analysis in
\Cref{lem:omitted-convolution-mass}
bounds the contribution of the excluded pairs.

\subsection{Open Problems}
A natural question is whether the time complexity of our FPTAS could be further improved to
near-quadratic or subquadratic time in $n$, and ultimately match the $n^{1.5}$ dependence of the
randomized algorithm. Another interesting question is whether the
$\widetilde O(n^{1.5}/\eps^2)$-time FPRAS of Feng and Jin \cite{FJ25} could be improved to
near-linear-time in $n$. On the other hand, it would be interesting to see if there is any barrier
for achieving $O(n^{2 - \Omega(1)}/\poly(\eps))$ time deterministically (similar to the works of
\cite{AbboudB17,AbboudR18,ChenGLRR19}).

\subsection{Paper Organization}
In \Cref{sec:preliminaries}, we give standard assumptions, definitions, and notation used throughout the
paper. In \Cref{sec:balanced-subsets}, we introduce essential lemmas for our deterministic balancing
approach, which constitute our main technical contribution. Finally, in \Cref{sec:fptas}, we
incorporate the tools from
\Cref{sec:balanced-subsets} into the partition-and-convolve framework of~\cite{GMW18} and obtain an improved
FPTAS, proving \Cref{thm:count-knapsack-fptas}.

\section{Preliminaries}
\label{sec:preliminaries}

\paragraph{Assumptions}
Let $w_1, w_2, \dots, w_n$ be the non-negative integer weights of the $n$ input items, and let $T$
be the input integer knapsack capacity.
We assume $w_i \le T$ for all $i$, since items with weights outside this range can be removed.
We assume $w_i > 0$; otherwise, remove the set $S$ of items with weight $0$ and multiply the final
answer by $2^{|S|}$.
We assume $w_1 + \cdots + w_n > T$; otherwise, every subset of items has total weight at most $T$
and we can return $2^n$.
We assume $\eps \in (0, 1/3)$; if $\eps \ge 1/3$, we may set $\eps$ to any fixed value smaller
than $1/3$.
We further assume that the $n$ weights are sorted in non-decreasing order,
$w_1 \le w_2 \le \cdots \le w_n$. Hence, for any non-empty $I \subseteq [n]$, the item
$\ell = \max I$ is the heaviest item in $I$.

\paragraph{Set Notation}
We write $[n] \coloneqq \{1, 2, \ldots, n\}$. For any reals $a < b$, we write $r \in (a, b)$ if
$a < r < b$ and write $r \in [a, b]$ if $a \le r \le b$.

\paragraph{Subset Weight and Feasible Subsets}
For a subset $S \subseteq [n]$, define its \emph{subset weight} by
\[
    w(S) \coloneqq \sum_{i \in S} w_i.
\]
For a subset $G \subseteq [n]$ and $t \ge 0$, let $\Omega_G[t]$ denote the collection of subsets of
$G$ with subset weight at most $t$. That is,
\[
    \Omega_G[t] \coloneqq \{\, I \subseteq G : w(I) \le t \,\}.
\]
We say that a subset $I \subseteq [n]$ is \emph{feasible} with respect to $t$ and $G$ if
$I \in \Omega_G[t]$. When the context is clear, we simply say that $I$ is \emph{feasible}, omitting
the dependence on $t$ and $G$. We denote the collection of feasible subsets with respect to $t$ and
$[n]$ as $\Omega[t]$ dropping the subscript. Our goal is to approximate $|\Omega[T]|$ efficiently.

\paragraph{Computational Model.}
We make the standard assumption that arithmetic operations
on the input integers take constant time. We analyze our
algorithm in a unit-cost arithmetic model, where arithmetic
operations and comparisons on intermediate count values
also take constant time. Thus, our running-time bounds count
arithmetic operations and do not necessarily imply the same
bit complexity.

\newcommand{\C}{\bigl\lceil 32c\sqrt{m \log m} + 2\bigr\rceil}
\newcommand{\D}{\bigl(32c\sqrt{m \log m} + 1\bigr)}
\newcommand{\E}{\bigl(32c\sqrt{m \log m}\bigr)}
\newcommand{\G}{\bigl\lceil 32c\sqrt{|G|\log |G|}+2 \bigr\rceil}
\section{Balanced Subsets}
\label{sec:balanced-subsets}
Our algorithm relies on a notion of \emph{balancedness} of subsets, which we discuss in this
section. Throughout this section, we fix a nonempty item group $G \subseteq [n]$ and a knapsack
capacity $t \ge 0$. Recall that we assume $w_1 \le w_2 \le \dots \le w_n$. We begin with the parity
partition of $G$.

\begin{definition}[Parity partition of $G$]
    \label{def:parity-partition}
    For an item group $G = \{i_1, i_2, \ldots, i_{|G|}\} \subseteq [n]$ with
    $i_1 < i_2 < \ldots < i_{|G|}$, we say $i_k \in G$ is \emph{even-indexed} if $k$ is even and
    \emph{odd-indexed} if $k$ is odd. Let $G_e$ be the set of even-indexed elements of $G$ and let
    $G_o$ be the set of odd-indexed elements of $G$. That is
    \[
        G_e := \{i_k \in G : k \text{ is even} \},
        \qquad
        G_o := \{i_k \in G : k \text{ is odd} \}.
    \]
\end{definition}
Clearly, we have $G = G_e \sqcup G_o$. We will show that for \emph{almost all} feasible subsets
$I \in \Omega_G[t] := \{I \subseteq G : w(I) \le t \}$, the even- and odd-indexed parts of $I$,
namely $I \cap G_e$ and $I \cap G_o$, are \emph{well balanced}—in the sense that removing a
\emph{small} number of the heaviest items from the heavier side makes its weight at most that of the
lighter side. Consequently, restricting attention to \emph{well-balanced feasible subsets} excludes
only a negligible fraction of feasible subsets, and we can gain computational savings in the merging
step of~\cite{GMW18} by convolving the two sum-approximations only over capacity ranges relevant to
well-balanced feasible subsets. We refer to \Cref{sec:fptas} for more details on the merging step
and focus here on formalizing balancedness.

As mentioned earlier in \Cref{sec:introduction}, in order to work without the bounded-ratio assumption on
input weights, it is more convenient to define the notion of balancedness based on item-removal
rather than the absolute weight difference. The definition is as follows:
\begin{definition}[$C$-balanced subsets with respect to $G$]
    \label{def:balanced-subsets}
    For $C > 0$, a subset $I \subseteq G$ is \emph{$C$-balanced with respect to $G$} if both of the
    following hold:
    \begin{itemize}
        \item There exists a subset $S \subseteq I \cap G_e$ of size $|S| \le C$ such that
            $w\bigl((I \cap G_e) \setminus S\bigr) \le w(I \cap G_o)$, and
        \item there exists a subset $S \subseteq I \cap G_o$ of size $|S| \le C$ such that
            $w\bigl((I \cap G_o) \setminus S\bigr) \le w(I \cap G_e)$,
    \end{itemize}
\end{definition}
Equivalently, $I$ is $C$-balanced with respect to $G$ if removing at most $C$ of the heaviest items
from the heavier side of $I \cap G_e$ and $I \cap G_o$ would make its remaining weight smaller than
or equal to the weight of the other side.

Let $c > 1$ be a parameter, and set $C := \G$. We will refer to $c$ as the balancing parameter. As
mentioned above, we restrict our attention to $C$-balanced feasible subsets of $G$ instead of
\emph{all} feasible subsets. To ensure that this restriction preserves the approximation accuracy,
we show that restricting to $C$-balanced feasible subsets excludes at most an $\cO(1/|G|^{c - 1})$
fraction of all feasible subsets. This claim is formalized in the following lemma.

\begin{lemma}[Balancing Lemma]
    \label{lem:deterministic-balancing}
    Let $t \ge 0$ be an arbitrary knapsack capacity and let $G \subseteq [n]$ be an arbitrary
    nonempty item group. Let $c > 1$ be an arbitrary balancing parameter and define $C := \G$. Let
    \[
        \begin{aligned}
            \Omega_G[t] & := \{I \subseteq G : w(I) \le t \}
            \qquad \text{and} \\
            \Gamma_G[t] & := \{I \in \Omega_G[t] : I \text{ is $C$-balanced with respect to } G \}.
        \end{aligned}
    \]
    Then
    \[
        |\Gamma_G[t]| \ge \bigl(1 - 8/|G|^{c - 1}\bigr)|\Omega_G[t]|.
    \]
\end{lemma}

To make our argument cleaner and avoid heavy notation, we reindex $G$ so that it becomes a
contiguous index set. Suppose $G = \{i_1, i_2, \ldots, i_m\}$ with $i_1 < i_2 < \ldots < i_m$. We
reindex $G$ as $G = [m] = \{1, 2, \ldots, m\}$ and set the item weights $w_k := w_{i_k}$
accordingly. Note that $w_1, \ldots, w_m$ are not necessarily the first $m$ items among the $n$
input items. However, the non-decreasing weight order $w_1 \le w_2 \le \cdots \le w_m$ and the even-
and odd-index parity of items in $G$ is preserved, so we may safely work with $G = [m]$. Note also
that
\[
    C := \C.
\]
With this reindexing, we obtain the following definitions and lemma:

\begin{definition}[Even/odd-indexed subset]
    For any subset $I \subseteq G = [m]$, define the even- and odd-indexed subsets of $I$ by
    \[
        I_e \coloneqq \{i \in I : i \text{ is even} \},
        \qquad
        I_o \coloneqq \{i \in I : i \text{ is odd} \}.
    \]
\end{definition}

\begin{definition}[$C$-balanced subset]
    \label{def:indexed-balanced-subsets}
    We say that a subset $I \subseteq G = [m]$ is \emph{$C$-balanced} if both of the following hold:
    \begin{itemize}
        \item There exists a subset $S \subseteq I_e$ of size $|S| \le C$ such that
            $w(I_e \setminus S) \le w(I_o)$,
        \item and, there exists a subset $S \subseteq I_o$ of size $|S| \le C$ such that
            $w(I_o \setminus S) \le w(I_e)$.
    \end{itemize}
\end{definition}

For all $I \subseteq G$, we have $I_e = I \cap G_e$ since $G_e = [m]_e$. Similarly, we have
$I_o = I \cap G_o$, so this definition is compatible with \cref{def:balanced-subsets}.

\begin{definition}[Feasible subsets]
    We denote the collection of feasible subsets by
    \[
        \Omega = \{I \subseteq [m] : w(I) \le t \}.
    \]
\end{definition}

\begin{definition}[$C$-balanced feasible subsets]
    \label{def:balanced-feasible-subsets}
    We denote the collection of $C$-balanced feasible subsets by
    \[
        \Gamma \coloneqq \bigl\{I \in \Omega : I \text{ is $C$-balanced} \bigr\}.
    \]
\end{definition}

With these definitions, \Cref{lem:deterministic-balancing} can be restated as follows.
\begin{restatable}[Balancing Lemma Reformulated]{lemma}{balancingreformulated}
    \label{lem:indexed-balancing}
    \[
        |\Gamma| \ge \bigl(1 - 8/m^{c - 1}\bigr)|\Omega|.
    \]
\end{restatable}

\subsection{Rounding Scheme \& Heaviest Item Removal}
To prove \Cref{lem:indexed-balancing}, it suffices to show that the feasible subsets in
$\Omega$ that are \emph{not} $C$-balanced constitute at most an $8/m^{c - 1}$ fraction of $\Omega$;
that is, $|\Omega \setminus \Gamma| \le 8/m^{c - 1} |\Omega|$. To this end, we introduce an
auxiliary rounding scheme on item weights and the induced rounded subset weights and rounded
feasible subsets, which we use to relate $|\Omega \setminus \Gamma|$ to $|\Omega|$.

\begin{definition}[Rounded item and subset weights]
    \label{def:rounded-subset-weights}
    For each $1 \le i \le m$, define the \emph{rounded item weight} $\widetilde w_i$ by
    \[
        \widetilde w_{2k - 1} \coloneqq \widetilde w_{2k} \coloneqq w_{2k}
        \qquad \text{for all integers $k \ge 1$ with $2k \le m$}.
    \]
    If $m$ is odd, then set $\widetilde w_m \coloneqq w_m$.
    For a nonempty subset $I \subseteq [m]$, define the \emph{rounded subset weight} by
    \[
        \widetilde w(I) \coloneqq w_\ell + \sum_{i \in I \setminus \{\ell\}} \widetilde w_i,
        \qquad \text{where } \ell \coloneqq \max I.
    \]
    Define $\widetilde w(\varnothing) = 0$.
\end{definition}

In other words, we round item weights by pairing consecutive indices $(2k - 1, 2k)$ and replacing
both weights with the heavier weight $w_{2k}$.
\begin{fact}
    \label{fact:rounding-only-increases-weights}
    For all $1 \le i \le m$, we have $\widetilde w_i \ge w_i$.
\end{fact}
\begin{proof}
    By construction, either $\widetilde w_i = w_i$ or $\widetilde w_i = w_{i + 1}$. Note that items
    are sorted in non-decreasing order $w_1 \le \ldots \le w_m$, so $\widetilde w_i \ge w_i$ in both
    cases.
\end{proof}
For the \emph{rounded subset weight} in \cref{def:rounded-subset-weights}, each $i \in I$ contributes its
rounded weight $\widetilde w_i$, except for the heaviest item $\ell = \max I$, which retains its
original weight $w_\ell$. We leave the heaviest item unrounded so that the discrepancy
$|\widetilde w(I) - w(I)|$ is bounded by $w_\ell$, as stated in the following fact:
\begin{fact}
    \label{fact:subset-rounding-error}
    For all nonempty subsets $I \subseteq [m]$, we have
    \[
        w(I) \le \widetilde w(I) < w(I) + \max_{i \in I} w_i.
    \]
\end{fact}
\begin{proof}
    The first inequality follows from \Cref{fact:rounding-only-increases-weights}. For the second, let $\ell = \max I$. Then
    \[
        \widetilde w(I) - w(I)
         = \left(\sum_{i \in I \setminus \{\ell\}} \widetilde w_i + w_\ell\right)
            - \sum_{i \in I} w_i
        = \sum_{i \in I \setminus \{\ell\}} (\widetilde w_i - w_i)
        \le \sum_{1 \le i < \ell} (\widetilde w_i - w_i).
    \]
    By \cref{fact:rounding-only-increases-weights} and using the non-decreasing order of $w_i$,
    \[
        \sum_{1 \le i < \ell} (\widetilde w_i - w_i)
        \le \sum_{1 \le i < \ell} (w_{i + 1} - w_i)
        = w_\ell - w_1
        < w_\ell.
    \]
    Hence, $\widetilde w(I) - w(I) < w_\ell = \max_{i \in I} w_i$.
\end{proof}

Based on the rounded subset weight defined in \cref{def:rounded-subset-weights}, we next define the
collection of rounded feasible subsets $\widetilde \Omega$, which will serve as an intermediate
object in bounding $|\Omega \setminus \Gamma|$.
\begin{definition}[Rounded feasible subsets]
    \[
        \widetilde \Omega \coloneqq \{I \subseteq [m] : \widetilde w(I) \le t \}.
    \]
\end{definition}
At a high level, our proof of \Cref{lem:indexed-balancing} proceeds as follows. We first show
that any feasible subset in $\Omega$ that is not $C$-balanced remains not $(C - 1)$-balanced after
removing its heaviest item (\Cref{lem:removal-preserves-imbalance}). Moreover, under the rounding scheme,
removing the heaviest item from a subset in $\Omega \setminus \widetilde\Omega$ produces a
\emph{rounded feasible subset} in $\widetilde\Omega$ (\Cref{fact:removal-map-feasibility}). We leave subsets
already in $\widetilde\Omega$ unchanged. Viewing this operation as a map, we can argue that each
image has at most $\cO(m)$ preimages. Hence, any rounded feasible subset in $\widetilde\Omega$ that
is not $(C - 1)$-balanced has at most $\cO(m)$ counterparts in $\Omega$ that are not $C$-balanced
(\Cref{fact:removal-map-preimages}). Finally, applying \Cref{lem:rounded-balancing-probability} with parameter $C - 1$, we show
that at most an $\cO(1/m^c)$ fraction of the subsets in $\widetilde\Omega$ are not
$(C - 1)$-balanced; combining these yields
\[
    |\Omega \setminus \Gamma|
    \le \cO\Bigl(m \cdot \bigl|\{J \in \widetilde\Omega : J \text{ is not } (C - 1)\text{-balanced}\}\bigr|\Bigr)
    \le \cO\left(\frac{m}{m^c}\right)|\widetilde\Omega|
    = \cO\left(\frac{1}{m^{c - 1}}\right)|\widetilde\Omega|.
\]
Since \Cref{fact:subset-rounding-error} implies that $|\widetilde\Omega| \le |\Omega|$, we obtain
\[
    |\Omega \setminus \Gamma|
    \le \cO\left(\frac{1}{m^{c - 1}}\right)|\Omega|.
\]

To make this argument precise, we define a map $\pi:\Omega \to \widetilde\Omega$ that removes the
heaviest item from a subset in $\Omega$ when needed. This map lets us carry out the analysis in the
$\widetilde\Omega$ space, as discussed above.

\begin{definition}[Heaviest item removal]
    Define $\pi:\Omega \to \widetilde\Omega$ by
    \[
        \pi(I) =
        \begin{cases}
            I, & \text{if } I \in \widetilde\Omega, \\
            I \setminus \{\max I\}, & \text{if } I \in \Omega \setminus \widetilde\Omega
        \end{cases}.
    \]
\end{definition}
Recall that $\max I$ is the heaviest item of $I$. Note that $\pi$ does \emph{not} remove $\max I$
from $I$ when $I \in \widetilde\Omega$, whereas it removes $\max I$ when
$I \notin \widetilde\Omega$.

\begin{fact}
    \label{fact:removal-map-feasibility}
    The heaviest item removal map $\pi$ is well-defined. That is, $\pi(I) \in \widetilde \Omega$ for
    all $I \in \Omega$.
\end{fact}
\begin{proof}
    It suffices to show $\pi(I) \in \widetilde \Omega$ when
    $I \in \Omega \setminus \widetilde \Omega$. Fix such an $I$. Since $I \notin \widetilde \Omega$,
    we have $I \neq \varnothing$. Let $\ell = \max I$. We claim that
    \[
        \widetilde w(I) - \widetilde w\bigl(\pi(I)\bigr) \ge w_{\ell}.
    \]
    Once this claim is proved, we have $\pi(I) \in \widetilde \Omega$ since
    \[
        \widetilde w\bigl(\pi(I)\bigr)
        \le \widetilde w(I) - w_{\ell}
        \le \bigl(w(I) + w_{\ell}\bigr) - w_{\ell}
        = w(I) \le t,
    \]
    where the second inequality follows from \Cref{fact:subset-rounding-error}, and the last follows from
    $I \in \Omega$.

    For the claim, if $|I| = 1$, then $\widetilde w(I) - \widetilde w\bigl(\pi(I)\bigr)$ equals the
    weight of the unique item in $I$, so the inequality holds trivially. Now assume $|I| \ge 2$, and
    let $\ell'$ be the second heaviest item in $I$. We have $\pi(I) = I \setminus \{\ell\}$, so
    \[
        \widetilde w(I) - \widetilde w\bigl(\pi(I)\bigr)
        = \left(w_{\ell} + \widetilde w_{\ell'}+ \sum_{i \in I \setminus \{\ell', \ell\}} \widetilde w_i\right)
        - \left(w_{\ell'} + \sum_{i \in I \setminus \{\ell', \ell\}} \widetilde w_i\right)
        = w_{\ell} + \widetilde w_{\ell'} - w_{\ell'}
        \ge w_{\ell},
    \]
    where the inequality uses \Cref{fact:rounding-only-increases-weights}.
\end{proof}

\subsection{Balancing Lemma}
As briefly mentioned earlier, we first observe that any feasible subset in $G$ that is not
$C$-balanced remains not $(C - 1)$-balanced under the heaviest-item removal map $\pi$:

\begin{lemma}
    \label{lem:removal-preserves-imbalance}
    \[
        \Omega \setminus \Gamma
        \subseteq \bigl\{I \in \Omega : \pi(I) \text{ is not } (C - 1)\text{-balanced} \bigr\}.
    \]
\end{lemma}
\begin{proof}
    Fix $I \in \Omega \setminus \Gamma$. By the definition of $\Gamma$ in
    \cref{def:balanced-feasible-subsets}, $I$ is not $C$-balanced, so clearly $I \neq \varnothing$. The goal
    is to prove $\pi(I)$ is not $(C - 1)$-balanced.

    Suppose for contradiction that $\pi(I)$ is $(C - 1)$-balanced. Then, by definition of
    balancedness (\cref{def:indexed-balanced-subsets}), there exists a subset $S \subseteq \pi(I)_e$ with
    $|S| \le C - 1$ such that $w\bigl(\pi(I)_e \setminus S\bigr) \le w\bigl(\pi(I)_o\bigr)$. Since
    $\pi(I)_e \subseteq I_e$ with $|I_e| - |\pi(I)_e| \le 1$ by definition of $\pi( \cdot )$, we
    know there exists $S' \subseteq I_e$ such that $I_e \setminus S' = \pi(I)_e \setminus S$ and
    $|S'| \le |S| + 1 \le C$, and consequently
    $w(I_e \setminus S') = w\bigl(\pi(I)_e \setminus S\bigr) \le w\bigl(\pi(I)_o\bigr) \le w(I_o)$.
    By a symmetric argument, there is also a subset $S'' \subseteq I_o$ of size $|S''| \le C$ such
    that $w(I_o \setminus S'') \le w(I_e)$. Together, this contradicts the assumption that $I$ is not
    $C$-balanced.

    Hence, $\pi(I)$ is not $(C - 1)$-balanced. This finishes the proof of the lemma.
\end{proof}

Next, we relate feasible subsets whose images under $\pi$ are not $(C - 1)$-balanced to rounded
feasible subsets that are not $(C - 1)$-balanced.

\begin{fact}
    \label{fact:removal-map-preimages}
    \[
            \bigl|\{I \in \Omega : \pi(I) \text{ is not } (C - 1)\text{-balanced} \}\bigr|
            \le (m + 1)
            \bigl|\{I \in \widetilde \Omega : I \text{ is not } (C - 1)\text{-balanced} \}\bigr|
    \]
\end{fact}
\begin{proof}
    For each $I \in \widetilde \Omega$ that is not $(C - 1)$-balanced, there are at most $m + 1$
    feasible subsets in $\Omega$ that map to $I$ under $\pi$: one such set is $I$ itself, and there
    is at most one additional such set for each possible choice of the heaviest item that could have
    been removed.
\end{proof}

Now, we show that at most an $\cO(1/m^{c - 1})$ fraction of the rounded feasible subsets in
$\widetilde\Omega$ are not $(C - 1)$-balanced.
\begin{lemma}
    \label{lem:rounded-balancing-probability}
    For any $c > 1$ and any integer $C \ge \D$, the probability that a uniformly random
    $J \in \widetilde \Omega$ is not $C$-balanced is at most $4/m^{c}$.
\end{lemma}
\begin{proof}
    Our proof for \cref{lem:rounded-balancing-probability} proceeds in the following steps:
    \begin{itemize}
        \item We first introduce a random modification procedure $\mathcal{M}$ which, given a
            $J' \in \widetilde \Omega$, returns a random
            $J = \mathcal{M}(J') \in \widetilde \Omega$.
        \item We show in \cref{claim:pair-swap-uniformity} that, if $J'$ is drawn from the uniform distribution
            over $\widetilde \Omega$, then the resulting $J = \mathcal{M}(J')$ remains uniformly
            distributed over $\widetilde \Omega$.
        \item Then we show that, for any fixed $J' \in \widetilde \Omega$, the probability (over the
            randomness of $\mathcal{M}$) that $\mathcal{M}(J')$ is not $C$-balanced is at most
            $4/m^c$.
        \item Combining the above two facts, we immediately conclude that a uniformly random member
            of $\widetilde \Omega$ is not $C$-balanced with probability at most $4/m^c$, finishing
            the proof.
    \end{itemize}

    Now we define the random modification procedure $\mathcal{M}$. Denote
    $K(\ell) = \lfloor (\ell - 1)/2\rfloor$ for short for any $\ell \in [m]$. Let
    $J' \in \widetilde \Omega$ be given. If $J' = \varnothing$, do nothing. Otherwise, for each
    $k \in [K(\max J')]$, do the following independently:
    \begin{enumerate}
        \item if exactly one of $\{2k - 1, 2k\}$ lies in $J'$, then with probability $1/2$ replace
            it with the other element in the pair;
        \item otherwise, either both or neither members of the pair $\{2k - 1, 2k\}$ lie in $J'$, in
            which case we do nothing.
    \end{enumerate}
    Return the resulting set as the output of $\mathcal{M}(J')$. When $J'$ is nonempty,
    $2K(\max J') < \max J'$, so we never modify its heaviest item $\max J'$.

    \begin{claim}
        \label{claim:pair-swap-uniformity}
        If $J'$ is drawn from the uniform distribution over $\widetilde \Omega$, then the set $J$
        returned by the random modification procedure $\mathcal{M}(J')$ is also uniformly
        distributed over $\widetilde \Omega$. Moreover, $|J| = |J'|$ and
        $\widetilde w(J) = \widetilde w(J')$ always hold.
    \end{claim}
    \begin{proof}
        Let $J$ be the output of this modification. First, it is clear that $|J| = |J'|$ since we
        replace items one for one. If $J' = \varnothing$, then $J = \varnothing$ and
        $\widetilde w(J) = \widetilde w(J') = 0$. Otherwise, by \Cref{def:rounded-subset-weights}, we
        have
        \[
            \widetilde w(J)
            = w_{\max J} + \sum_{j \in J \setminus \{\max J\}} \widetilde w_j
            = w_{\max J'} + \sum_{j \in J' \setminus \{\max J'\}} \widetilde w_j
            = \widetilde w(J')
        \]
        where the second equality follows from the fact that
        $\widetilde w_{2k - 1} = \widetilde w_{2k}$ and we did not modify the heaviest item in $J'$.
        Consequently, we also have $J \in \widetilde \Omega$.

        It remains to verify that the random modification preserves the uniform distribution on
        $\widetilde\Omega$. For any nonempty $X \in \widetilde \Omega$ and each
        $k \in [K(\max X) + 1]$, define $c_k(X) \in \{0, 1, 2\}$ as the number of elements of the
        pair $(2k - 1, 2k)$ contained in $X$. Then, for each $\ell \in [m]$ and each sequence
        $\mathbf{c} \in \{0, 1, 2\}^{K(\ell) + 1}$, define the equivalence class
        \[
            \mathcal{C}_{\ell, \mathbf{c}} =
            \bigl\{X \in \widetilde\Omega :
            \max X = \ell \ \text{and}\ c_k(X) = \mathbf{c}_k
            \ \text{for all } k \in [K(\ell) + 1]
            \bigr\}.
        \]
        The nonempty classes $\mathcal{C}_{\ell, \mathbf{c}}$ partition
        $\widetilde \Omega \setminus \{\varnothing\}$, and $\{\varnothing\}$ forms a separate class.
        Since $J'$ is uniform on $\widetilde\Omega$, this means
        $\Pr[J' \in \mathcal{C}_{\ell, \mathbf{c}}] = |\mathcal{C}_{\ell, \mathbf{c}}|/|\widetilde\Omega|$.
        Moreover, by definition of the random modification procedure $\mathcal{M}(J')$, it uniformly
        randomly toggles those pairs $(2k - 1, 2k)$ with $c_k(J') = 1$ and $k \in [ K(\max J')]$,
        leaving $\max J'$ and $c_k(J')$ unchanged for all $k \in [K(\max J') + 1]$; hence, it maps
        $J'$ to a uniformly random member of the same class containing $J'$. In other words,
        conditioned on $J' \in \mathcal{C}_{\ell, \mathbf{c}}$, $\mathcal{M}(J')$ returns a
        uniformly random $J \in \mathcal{C}_{\ell, \mathbf{c}}$. Therefore, for any nonempty
        $X \in \widetilde\Omega$,
        \[
            \Pr[J = X]
            = \sum_{\ell, \mathbf{c}: X \in \mathcal{C}_{\ell, \mathbf{c}}}
            \Pr\bigl[J' \in \mathcal{C}_{\ell, \mathbf{c}}\bigr]
            \Pr\bigl[J = X \mid J' \in \mathcal{C}_{\ell, \mathbf{c}}\bigr]
            = \sum_{\ell, \mathbf{c}: X \in \mathcal{C}_{\ell, \mathbf{c}}}
            \frac{|\mathcal{C}_{\ell, \mathbf{c}}|}{|\widetilde\Omega|}
            \cdot \frac{1}{|\mathcal{C}_{\ell, \mathbf{c}}|}
             = \frac{1}{|\widetilde\Omega|},
        \]
        where the first and the last equalities follow from the fact that the nonempty classes
        $\mathcal{C}_{\ell, \mathbf{c}}$ partition $\widetilde \Omega \setminus \{\varnothing\}$.
        Since $\mathcal{M}$ preserves cardinality, we also have
        $\Pr[J = \varnothing] = \Pr[J' = \varnothing] = 1/|\widetilde\Omega|$. Hence, $J$ is uniform
        on $\widetilde\Omega$.
    \end{proof}

    It remains to show that the modified rounded feasible subset is not $C$-balanced with small
    probability. Fix an integer $C \ge \D$. Given a fixed $J' \in \widetilde\Omega$, run the random
    modification procedure $\mathcal{M}(J')$ to obtain $J$. We may assume that $|J'| = |J| > C$, as
    otherwise $J$ is trivially $C$-balanced. We consider the case $w(J_e) > w(J_o)$; the other case
    is analogous. We introduce some definitions needed for our analysis:
    \begin{itemize}

        \item \textbf{Largest even index $\ell_e$.} Let
            $\ell_e \coloneqq 2\lfloor (\max J' - 1)/2\rfloor$. By definition of our procedure
            $\mathcal{M}$, indices in $J'$ larger than $\ell_e$ are not modified and so they are
            always contained in $J$.

        \item \textbf{Randomized/deterministic region $R$ and $D$.} Define $R \subseteq [m]$ to be
            the set of indices $2k - 1$ and $2k$ whenever $c_k(J') = 1$ and $k \in [\ell_e/2]$.
            Intuitively, $R$ consists of those indices that may be modified by the random procedure.
            Define $D \coloneqq [m] \setminus R$. Note that $D$ is the set of indices $2k - 1$ and
            $2k$ with $k \in [\ell_e/2]$ and either $c_k(J') = 0$ or $c_k(J') = 2$, as well as all
            indices that exceed $\ell_e$. Moreover, indices in $D$ are untouched by the random
            modification, so $J \cap D = J' \cap D$.

        \item \textbf{Threshold $h$ and Heavy/Light Items.} Let $H' \subseteq J'$ be the set of the
            $C - 1$ heaviest items in $J'$ (ties broken by largest index) and define
            $h := 2\lceil \min H'/2\rceil - 1$. Note that $|H'| = C - 1$, since $|J'| = |J| > C$.
            Then, $|J' \cap [h, m]| \le |H'| + 1 = C$.

        For brevity, let
        \[
            J_{[h, m]} := [h, m] \cap J \qquad \text{and} \qquad J_{[1, h - 1]} := [1, h - 1] \cap J.
        \]
        Since the weights are in non-decreasing order, $ w_j \ge w_h$ for all $j \in J_{[h, m]}$ and
        $w_j \le w_h$ for all $j \in J_{[1, h - 1]}$. Also, since $h$ is odd, by definition of the
        random modification procedure it holds that $|J_{[h, m]}| = |J' \cap [h, m]| \le C$.
        Furthermore, define the set of even indices
        \[
            J^{e}_{[1, h - 1]} := [1, h - 1] \cap J_e \qquad \text{ and } \qquad J^{e}_{[h, m]} := [h, m] \cap J_e
        \]
        and the odd indices
        \[
            J^{o}_{[1, h - 1]} := [1, h - 1] \cap J_o \qquad \text{ and } \qquad J^{o}_{[h, m]} := [h, m] \cap J_o
        \]
    \end{itemize}

    Now, we remove the set $J^{e}_{[h, m]} \subseteq J_e$ from $J_e$ and show that
    \[
        w\bigl(J_e \setminus J_{[h, m]}^{e}\bigr) \le w(J_o)
    \]
    with probability at least $1 - 2/m^c$ over the random modification procedure. Since
    $J_o = J^{o}_{[h, m]} \sqcup J^{o}_{[1, h - 1]}$ and
    $J_e = J^{e}_{[h, m]} \sqcup J^{e}_{[1, h - 1]}$, the above inequality is equivalent to
    \begin{equation}
        \label{eq:heavy-light-balancing-target}
        \sum_{j \in J^{e}_{[1, h - 1]}} w_j - \sum_{j \in J^{o}_{[1, h - 1]}} w_j \le \sum_{j \in J^{o}_{[h, m]}} w_j.
    \end{equation}
    Note that $|J^{e}_{[h, m]}| \le |J_{[h, m]}| \le C$, so proving \Cref{eq:heavy-light-balancing-target} with
    probability at least $1 - 2/m^c$ suffices. We will prove \cref{eq:heavy-light-balancing-target} by separately bounding
    the left-hand side and the right-hand side.

    \begin{claim}
        \label{claim:light-part-discrepancy}
        \[
            \sum_{j \in J^{e}_{[1, h - 1]}} w_j - \sum_{j \in J^{o}_{[1, h - 1]}} w_j
                < \Bigl(1 + \frac{|H'|}{16}\Bigr)w_h
        \]
        with probability at least $1 - m^{- c}$.
    \end{claim}
    \begin{proof}
        We rewrite the difference in terms of sums over $R$ and $D$ as
        \[
            \sum_{j \in J^{e}_{[1, h - 1]}} w_j - \sum_{j \in J^{o}_{[1, h - 1]}} w_j
            = \sum_{j \in J_{[1, h - 1]}} ( - 1)^j w_j
            = \underbrace{\sum_{j \in J_{[1, h - 1]} \cap D} ( - 1)^j w_j}_{S_D}
            + \underbrace{\sum_{j \in J_{[1, h - 1]} \cap R} ( - 1)^j w_j}_{S_R}.
        \]
        Recall $h - 1$ is even. To bound $S_D$, recalling the definition of $D$, we know that for
        every pair $\{2k - 1, 2k\}$, either both elements are contained in $J_{[1, h - 1]} \cap D$
        and contribute $w_{2k} - w_{2k - 1}$ to $S_D$, or neither is contained, contributing $0$ to
        $S_D$. Hence, by the monotonicity of the weights (using a similar argument to that in
        \cref{fact:subset-rounding-error}),
        \begin{align*}
            S_D
            \le \sum_{\{2k - 1, 2k\} \subseteq J_{[1, h - 1]} \cap D} \bigl(w_{2k} - w_{2k - 1}\bigr)
            \le \sum_{\substack{1 \le j < h \\ j\ \text{odd}}} (w_{j + 1} - w_j)
            & \le \sum_{\substack{1 \le j < h}} (w_{j + 1} - w_j) \\
            &= w_h - w_1 < w_h.
        \end{align*}

        Now we bound $S_R$. By definition of $R$ and our random modification procedure,
        $J_{[1, h - 1]} \cap R$ independently includes exactly one random element from every pair
        $\{2k - 1, 2k\}$ where $k \in K' := \{k \in [m] : 2k < h, \text{ and } c_k(J') = 1\}$. Thus,
        we can write $S_R = \sum_{k \in K'} X_k$, where $X_k$ are random variables distributed as
        \[
            X_k =
            \begin{cases}
                w_{2k} & \text{ with probability $1/2$} \\
                - w_{2k - 1} & \text{ with probability $1/2$}
            \end{cases}
        \]
        independently across $k \in K'$. Then,
        \begin{align*}
            \mathbb{E}[S_R]
            = \sum_{k \in K'} \mathbb{E}[X_k]
            = \frac{1}{2} \sum_{k \in K'} (w_{2k} - w_{2k - 1})
            \le \frac{1}{2} \sum_{2k < h} \bigl(w_{2k} - w_{2k - 1}\bigr)
            & \le \frac{1}{2} \sum_{\substack{1 \le j < h \\ j\ \text{odd}}}
            \bigl(w_{j + 1} - w_{j}\bigr)\\
            &< \frac{w_h}{2}.
        \end{align*}
        Since $w_j \le w_h$ for all $j \in J_{[1, h - 1]}$, we have $X_k \in [ - w_h, w_h]$. If
        $K' = \varnothing$, then $S_R = 0$ and the desired bound holds trivially. Otherwise, by
        Hoeffding's inequality,
        \begin{align*}
            \Pr\Bigl[S_R > \frac{|H'|}{16}w_h\Bigr]
            &= \Pr\Bigl[
                S_R-\mathbb E[S_R]
                > \frac{|H'|}{16}w_h-\mathbb E[S_R]
            \Bigr] \\
            &\le \Pr\Bigl[
                S_R-\mathbb E[S_R]
                > \frac{|H'|-8}{16}w_h
            \Bigr] \\
            &\le \exp\left(
                -\frac{2\bigl(w_h(|H'|-8)/16\bigr)^2}
                    {4|K'|w_h^2}
            \right).
        \end{align*}
        where we used $|H'| = C - 1 \ge 32c\sqrt{m\log m} > 8 + \sqrt{2} \cdot 16c \sqrt{m \log m}$.
        Therefore, with probability at least $1 - m^{- c}$,
        \[
            S_D + S_R
            < w_h + \frac{|H'|}{16} w_h
            = \Bigl(1 + \frac{|H'|}{16}\Bigr) w_h.
        \]
        This finishes the proof of \cref{claim:light-part-discrepancy}.
    \end{proof}

    \begin{claim}
        \[
            \sum_{j \in J^{o}_{[h, m]}} w_j \ge \Bigl(1 + \frac{|H'|}{16}\Bigr) w_h
        \]
        with probability at least $1 - m^{- c}$.
    \end{claim}
    \begin{proof}
        Using the fact that $w_j \ge w_h$ holds for all $j \in [h, m]$, it suffices to show that
        $|J_{[h, m]}^o| \ge 1 + |H'|/16$ holds with probability at least $1 - m^{- c}$.

        Let
        \[
            K'' := \bigl\{k \in \mathbb{Z} : h < 2k \le \ell_e \text{ and } c_k(J') \neq 0 \bigr\}.
        \]
        Recalling the definitions of $h, \ell_e$, we have
        $2|K''| \ge |J' \cap [h, \ell_e]| \ge |J' \cap [h, m]| - 2 \ge |H'| - 2$. By definition of
        the random modification procedure, independently for every $k \in K''$, the event
        $2k - 1 \in J^o_{[h, m]}$ happens either with probability $1/2$ if $c_k(J') = 1$ , or with
        probability $1$ if $c_k(J') = 2$. Hence, the random variable $|J^o_{[h, m]}|$ stochastically
        dominates the binomial distribution $\mathrm{Bin}(|K''|, 1/2)$. Therefore, by Chernoff's
        bound:
        \begin{align*}
            \Pr \bigl[ |J^o_{[h, m]}| < 1 + |H'|/16 \bigr]
            \le \Pr \bigl[ |J^o_{[h, m]}| \le (|H'| - 2) /8 \bigr]
            & \le \Pr \bigl[ |J^o_{[h, m]}| \le |K''| /4 \bigr] \\
            & \le \exp \bigl(- |K''|/8\bigr) \\
            & \le \exp \bigl(- (|H'| - 2)/16\bigr) \\
            & \le \exp \bigl(- c\sqrt{m \log m}\bigr) \\
            & \le \exp \bigl(- c \log m\bigr) \\
            & = m^{- c},
        \end{align*}
        where we used $|H'| = C - 1 \ge \E$. Therefore, with probability at least $1 - m^{- c}$, the
        claimed bound holds.
    \end{proof}

    Combining the bounds from the two claims, we obtain \Cref{eq:heavy-light-balancing-target} with probability at least
    $1 - 2m^{- c}$:
    \[
        \sum_{j \in J^{e}_{[1, h - 1]}} w_j - \sum_{j \in J^{o}_{[1, h - 1]}} w_j
        < \Bigl(1 + \frac{|H'|}{16}\Bigr) w_h
        < \sum_{j \in J^{o}_{[h, m]}} w_j
    \]
    as desired. The case $w(J_o) > w(J_e)$ is analogous. Hence, by a union bound, we conclude that
    $J$ is $C$-balanced with probability at least $1 - 4m^{- c}$. By the discussion at the beginning
    of the proof of \cref{lem:rounded-balancing-probability}, this proof is complete.
\end{proof}

Finally, we combine \Cref{lem:removal-preserves-imbalance}, \Cref{fact:removal-map-preimages}, and
\Cref{lem:rounded-balancing-probability} to prove \Cref{lem:indexed-balancing}.

\balancingreformulated*
\begin{proof}
    We show that $|\Omega \setminus \Gamma| \le 8|\Omega|/m^{c - 1}$. By
    \Cref{lem:removal-preserves-imbalance} and \Cref{fact:removal-map-preimages},
    \begin{align*}
        |\Omega \setminus \Gamma|
        \le \bigl|\{I \in \Omega : \pi(I) \text{ is not } (C - 1)\text{-balanced} \}\bigr|
        \le (m + 1) \bigl|\{I \in \widetilde \Omega : I \text{ is not } (C - 1)\text{-balanced} \}\bigr|.
    \end{align*}
    Recall that $(C - 1) = (\G - 1) \ge 32c\sqrt{|G|\log |G|} + 1$. By \Cref{lem:rounded-balancing-probability}, for a
    uniformly random $X \in \widetilde{\Omega}$, the probability that $X$ is not $(C - 1)$-balanced
    is at most $4/m^{c}$. Hence,
    \[
        \begin{aligned}
            |\Omega \setminus \Gamma|
            \le (m + 1) \bigl|\{I \in \widetilde \Omega : I \text{ is not } (C - 1)\text{-balanced} \}\bigr|
            \le \frac{4(m + 1)}{m^c}|\widetilde \Omega|
            \le \frac{8}{m^{c - 1}}|\widetilde \Omega|.
        \end{aligned}
    \]
    By \Cref{fact:subset-rounding-error}, we conclude
    \[
        |\Omega \setminus \Gamma|
        \le \frac{8}{m^{c - 1}}|\widetilde \Omega|
        \le \frac{8}{m^{c - 1}} |\Omega|.
    \]
\end{proof}

\section{FPTAS for \#Knapsack}
\label{sec:fptas}
In this section, we discuss our FPTAS for \#Knapsack. For an integer $t \ge 0$, recall that
$\Omega[t] := \{I \subseteq [n] : w(I) \le t \}$ is defined in \Cref{sec:preliminaries}. Our goal here is to
approximate $|\Omega[T]|$ within a factor of $(1 \pm \eps)$ in time $\tcO(n^{2.4}/\eps^{1.4})$.
Before presenting our algorithm, we define some count functions and briefly summarize the
partition-and-convolve framework of~\cite{GMW18} on which our algorithm and analysis rely.

\subsection{Count and Approximation Function}
\begin{definition}[Support of a function]
    For a function $f \colon \mathbb{Z}_{\ge 0} \to \mathbb{Z}_{\ge 0}$, the support of $f$ is
    \[
        \supp(f) \coloneqq \bigl\{x \in \mathbb{Z}_{\ge 0} : f(x) \ne 0\bigr\}.
    \]
\end{definition}
Note that the support of $f$ is the subset of its domain where $f$ is nonzero. We refer to the
\emph{size} of $f$ as $|\supp(f)|$; this is natural since we need only to store the pairs
$(x, f(x))$ with $f(x) > 0$. We assume that $\supp(f)$ is always sorted.

\begin{definition}[Count function]
    \label{def:count-function}
    For any item group $G \subseteq [n]$, the count function of $G$ is
    $f_G \colon \mathbb{Z}_{\ge 0} \to \mathbb{Z}_{\ge 0}$, defined
    \[
        f_G(t) \coloneqq \bigl|\{I \subseteq G : w(I) = t \}\bigr|.
    \]
\end{definition}
In other words, $f_G(t)$ is the number of subsets of $G$ whose total weight is \emph{exactly} $t$.
When $G = [n]$, we denote its count function simply by $h$ to distinguish it from others.

\begin{definition}[Prefix sum]
    For a function $f \colon \mathbb{Z}_{\ge 0} \to \mathbb{Z}_{\ge 0}$, the prefix sum of $f$ is
    \[
        f^{\le}(t) \coloneqq \sum_{0 \le x \le t} f(x).
    \]
\end{definition}
Unless otherwise specified, a superscript $\le$ denotes the prefix sum. Note that if $f_G$ is the
count function of $G \subseteq [n]$, then $f_G^{\le}(t) = |\Omega_G[t]|$. Hence, $h^{\le}(T)$ equals
the \emph{exact} number of feasible subsets of $[n]$ (i.e., $h^{\le}(T) = |\Omega[T]|$) and we want
to approximate it efficiently.

\begin{definition}[$(1 + \eps, 1 - \gamma)$-sum-approximation]
    \label{def:sum-approximation}
    For $\eps, \gamma \ge 0$, a function $\tilde f \colon \mathbb{Z}_{\ge 0} \to \mathbb{Z}_{\ge 0}$
    is a $(1 + \eps, 1 - \gamma)$-sum-approximation of a function
    $f \colon \mathbb{Z}_{\ge 0} \to \mathbb{Z}_{\ge 0}$ if
    \[
        (1 - \gamma) f^{\le}(x) \le \widetilde f^{\le}(x) \le (1 + \eps) f^{\le}(x)
    \]
    for all $x \in \mathbb{Z}_{\ge 0}$.
\end{definition}
Note that $\widetilde f(x)$ does not necessarily approximate $f(x)$ but its prefix sum approximates
the prefix sum of $f$. Unless otherwise specified, we use a tilde to indicate a sum-approximation.
When $\gamma = 0$, we simply say that $\tilde f$ is a $(1 + \eps)$-sum-approximation of $f$. When
$\eps = \gamma$, we simply say that $\widetilde f$ is a $(1 \pm \eps)$-sum-approximation of $f$.

Note that if $\widetilde h$ is a $(1 \pm \eps)$-sum-approximation of the count function $h$, then
\[
    (1 - \eps)|\Omega[t]| = (1 - \eps)h^{\le}(t) \le \widetilde h^{\le}(t) \le (1 + \eps) h^{\le}(t)
        = (1 + \eps) |\Omega[t]| \qquad \forall t \ge 0.
\]
Hence, the value \(\widetilde h^{\le}(T)\) is a \((1 \pm \eps)\)-approximation to \(|\Omega[T]|\).

\begin{definition}[Function convolution]
    Given two functions $f \colon \mathbb{Z}_{\ge 0} \to \mathbb{Z}_{\ge 0}$ and
    $g \colon \mathbb{Z}_{\ge 0}\to \mathbb{Z}_{\ge 0}$, their convolution is the function
    $f \star g \colon \mathbb{Z}_{\ge 0}\to \mathbb{Z}_{\ge 0}$ defined by
    \[
        (f \star g)(t) = \sum_{x + y = t} f(x) g(y).
    \]
\end{definition}
The running time to compute $(f \star g)(t)$ for all $t \ge 0$ is
$\tcO(|\supp(f)| \cdot |\supp(g)|)$: for each $x \in \supp(f)$ and $y \in \supp(g)$, add the nonzero
product $f(x)g(y)$ to the entry at $t = x + y$ in the representation of $f \star g$.

\subsection{Partition-and-Convolve Framework}
The high-level idea in~\cite{GMW18} for computing $\widetilde h$ is as follows. Recursively
partition the items $[n]$ into halves until we obtain a full binary recursion tree of depth $\log p$
with $p$ leaf groups. For each leaf group, we apply a simplified version of Halman’s algorithm
(\Cref{lem:halman-leaf-algorithm}) to compute a sum-approximation with small support size\footnote{More precisely,
the size is $\widetilde O(|G|/\delta)$, where $G$ is the item set at that recursion step and
$\delta$ is the sparsification parameter.}. Then, proceeding bottom-up, we repeatedly convolve the
sum-approximations (\Cref{lem:full-convolution-approximation}) corresponding to the two children of a recursion node to
obtain a sum-approximation for the parent node. However, the brute-force convolution is quadratic in
the support sizes, and each convolution produces an output whose support size is quadratic in those
of the inputs. Thus, after each convolution, they sparsify the result (\Cref{lem:sparsify-count-function}) to
keep the support size small, at the cost of a controlled approximation error. Repeating this
convolution–sparsification step up the recursion tree eventually gives a sum-approximation for $h$
at the root, corresponding to $[n]$. With carefully chosen sparsification parameters, this yields an
overall running time of $\tcO\bigl(n^{2.5}/\eps^{1.5}\bigr)$ while preserving a final $(1 + \eps)$
relative error. More precisely, they prove the following three lemmas, which we also use.

\begin{lemma}[Sparsification of a sum-approximation in~\cite{GMW18} (Section~2)]
    \label{lem:sparsify-count-function}
    Let $G \subseteq [n]$, and let $f_G$ be its count function. Suppose we are given a
    $((1 + \eps), (1 - \gamma))$-sum-approximation $\widetilde f$ of $f_G$. Then, for any
    $\delta > 0$, there exists an algorithm that runs in $\tcO(|\supp(\widetilde f)| + |G|/\delta)$
    time and outputs a function $\widetilde g \colon \mathbb{Z}_{\ge 0} \to \mathbb{Z}_{\ge 0}$ such
    that
    \begin{enumerate}
        \item $\widetilde g$ is a $((1 + \delta)(1 + \eps), (1 - \gamma))$-sum-approximation of
            $f_G$,
        \label{item:sparsification-approximation-factor}
        \item the support of $\widetilde g$ has size $\widetilde O(|G|/\delta)$,
        \label{item:sparsification-support-size}
        \item if $\supp(\widetilde g) = \{b_1, \ldots, b_m\}$ with $b_1 < \cdots < b_m$, then
        \[
            \widetilde g^{\le}(b_i) \ge (1 + \delta) \widetilde g^{\le}(b_{i - 1}), \qquad \forall 2
                \le i \le m.
        \]
        \label{item:sparsification-geometric-growth}
    \end{enumerate}
\end{lemma}

We refer to $\delta$ as the \emph{sparsification parameter}. We also call the algorithm’s output the
\emph{$\delta$-sparsified} $((1 + \delta)(1 + \eps), (1 - \gamma))$-sum-approximation of $f_G$.
Intuitively, $\delta$-sparsification reduces the support size at the cost of an extra multiplicative
factor $(1 + \delta)$ by discretizing the function: the image of $\widetilde g^{\le}$ is divided
into geometric levels $(1 + \delta)^0, (1 + \delta)^1, \ldots, (1 + \delta)^k$, which in turn
determine the elements in $\supp(\widetilde g)$. To achieve this, the algorithm scans
$\supp(\widetilde f)$ and selects the first capacity-index at which $\widetilde f^{\le}$ reaches
each successive power of $(1 + \delta)$. Since $G$ has at most $2^{|G|}$ feasible subsets, the
resulting function has support size
\[
    \log_{1 + \delta} 2^{|G|} = \tcO\left(\frac{|G|}{\delta}\right),
\]
which gives \Cref{item:sparsification-support-size}. This construction also immediately implies
\Cref{item:sparsification-approximation-factor,item:sparsification-geometric-growth}.

\begin{lemma}[Simplified Halman's algorithm
in~\cite{GMW18} (Section~3)]
    \label{lem:halman-leaf-algorithm}
    Given $\eps\in(0,1/3)$ and $n$ weights, there is an
    algorithm that computes an $(\eps/3)$-sparsified
    $(1\pm\eps)$-sum-approximation of $h$ with support
    size $\widetilde O(n/\eps)$ in $O(n^3/\eps)$ time.
    \footnote{
    Run the simplified Halman algorithm of~\cite{GMW18}
    with error parameter $\eps/3$, and then apply
    $(\eps/3)$-sparsification using
    \Cref{lem:sparsify-count-function}.
    The resulting upper approximation factor is at most
    $(1+\eps/3)^2\le 1+\eps$, while the lower factor is
    at least $1-\eps$.
    Both steps take $O(n^3/\eps)$ total time, and the
    final support size is $\widetilde O(n/\eps)$.
    }
\end{lemma}

We view this algorithm as a modification of the standard dynamic program for \#Knapsack: define
\[
    d(i, t) := \bigl|\{I \subseteq [i] : w(I) \le t \}\bigr|.
\]
One can verify that $d$ satisfies the recurrence
\[
    d(i, t) = d(i - 1, t) + d(i - 1, t - w_i), \qquad \forall 1 \le i \le n, \ \forall 0 \le t \le
        T,
\]

Note that the array $d(i, \cdot)$ has size $T + 1$ entries, which is too large.
Hence, for $i=1,\ldots,n$, we maintain a $\delta$-sparsified approximation of the exact-weight
count function, whose prefix sums approximate $d(i,\cdot)$. The maintained array has size
$\widetilde O(i/\delta)$, with each sparsification step incurring a multiplicative
$(1+\delta)$ error in the prefix counts. Taking $\delta = (1 + \eps)^{1/n} - 1$ makes the error 
over $n$ iterations compound to $((1 + \eps)^{1/n})^n = 1 + \eps$, and ensures that
\[
    \# \text{stored breakpoints at step $i$} = \widetilde O\left(\frac{i}{\delta}\right) =
        \widetilde O\left(\frac{in}{\eps}\right).
\]
In particular, for $i \le n$, this is $\widetilde O(n^2/\eps)$. Thus, computing $d(i, \cdot)$ at
step $i$ takes $\widetilde O(n^2/\eps)$ time, yielding a total running time
$\widetilde O(n^3/\eps)$.

\begin{lemma}[Merging in~\cite{GMW18} (Lemmas 1 and 4)]
    \label{lem:full-convolution-approximation}
    Let $A, B \subseteq [n]$ with $A \cap B = \varnothing$, and let
    $\eps_A, \eps_B, \gamma_A, \gamma_B > 0$. Suppose $\widetilde f_A$ is a
    $((1 + \eps_A), (1 - \gamma_A))$-sum-approximation of the counting function $f_A$, and
    $\widetilde f_B$ is a $((1 + \eps_B), (1 - \gamma_B))$-sum-approximation of the counting
    function $f_B$. Then $\widetilde f_A \star \widetilde f_B$ is a
    $((1 + \eps_A)(1 + \eps_B), (1 - \gamma_A)(1 - \gamma_B))$-sum-approximation of the counting
    function $f_{A\cup B}$.
\end{lemma}
Intuitively, the lemma is a direct product of a property of the count function:
\[
    f_{A\cup B}(t) = \sum_{x + y = t} f_A(x)f_B(y) = (f_A \star f_B)(t),
\]
which implies
\[
    f_{A\cup B}^{\le}(t) = \sum_{x + y \le t} f_A(x)f_B(y) = (f_A \star f_B)^{\le}(t).
\]
Thus, when we replace $f_A, f_B$ with $\widetilde f_A, \widetilde f_B$, the upper multiplicative
errors $(1 + \eps_A)$ and $(1 + \eps_B)$ multiply, and so do the lower factors $(1 - \gamma_A)$ and
$(1 - \gamma_B)$, giving the claimed approximation guarantee.

Putting the three lemmas together, after computing sum-approximations for the $p$ leaves,
\cite{GMW18} convolve the two child sum-approximations at each internal node and then
$\delta_{\ell}$-sparsify the resulting function, maintaining a $\delta_{\ell}$-sparsified
sum-approximation at every node on level $\ell$. Our algorithm follows the same template; the key
difference is that we use a different partition and perform convolution only over a small window
rather than over the full range.

\subsection{The Algorithm}
We use the partition-and-convolve framework of~\cite{GMW18}, together with our notion of balanced
subsets from \Cref{sec:balanced-subsets}, to compute a $(1 \pm \eps)$-sum-approximation $\widetilde h$ of
$h$ efficiently. First, if $\eps < n^{-3/2}$, we apply \Cref{lem:halman-leaf-algorithm} directly. Its running time
satisfies
\[
    \cO(n^3/\eps) \le O(n^{12/5}/\eps^{7/5}),
\]
as $n^{3/5}\eps^{2/5} < 1$. Thus, for the rest of this subsection, we assume $\eps \ge n^{-3/2}$.

We partition the item sets so that \Cref{lem:deterministic-balancing} applies. Fix a parameter $1 \le p \le n$
that is a power of two; we choose its value in the running-time analysis. Let
$L \coloneqq \log_2 p$. For $\ell = 0, \ldots, L$, our algorithm constructs a recursive partition
$\{G_1^{\ell}, G_2^{\ell}, \ldots, G_{2^{\ell}}^{\ell}\}$ of
$[n] = \bigsqcup_{k = 1}^{2^{\ell}} G_k^{\ell}$ as follows:

\begin{enumerate}
    \item Initialization ($\ell = 0$): $G_1^{0} \coloneqq [n]$.
    \item Inductive step ($1 \le \ell \le L$): For each $1 \le k \le 2^{\ell - 1}$, write
        $G_k^{\ell - 1} = \{i_1, \ldots, i_m \}$ in increasing order and partition it based on their
        parity as in \Cref{def:parity-partition}:
    \[
        G_{2k - 1}^{\ell} := \bigl(G_k^{\ell - 1}\bigr)_o = \{i_j \in G_k^{\ell - 1} : j \text{ is
            odd} \}
    \]
    and
    \[
        G_{2k}^{\ell} := \bigl(G_k^{\ell - 1}\bigr)_e = \{i_j \in G_k^{\ell - 1} : j \text{ is even}
            \}.
    \]
\end{enumerate}
This yields a full binary recursion tree of depth $L$ with $p$ leaves, where the $k$th node at level
$\ell$ corresponds to the item group $G_k^{\ell}$ and is partitioned as
\[
    G_k^{\ell} = G_{2k - 1}^{\ell + 1} \sqcup G_{2k}^{\ell + 1}, \qquad \forall 0 \le \ell \le L -
        1, \ 1 \le k \le 2^{\ell}.
\]
Moreover, $|G_k^{\ell}| \in \{\lfloor n/2^{\ell} \rfloor, \lceil n/2^{\ell} \rceil\}$, so in
particular $|G_k^{\ell}| = \Theta(n/2^{\ell})$.

For all $0 \le \ell \le L$ and $1 \le k \le 2^{\ell}$, let $f_k^{\ell}$ be the count function of
$G_k^{\ell}$ as in \Cref{def:count-function}. We set the sparsification parameter at recursion level
$\ell$ to
\[
    \delta_{\ell} := \frac{\eps}{2dp^{1/2}2^{\ell/2}},
    \qquad \text{where} \quad d := \frac{\sqrt{2}}{\sqrt{2} - 1}.
\]
Proceeding bottom-up along the recursion tree, we compute a $\delta_{\ell}$-sparsified
$(1 \pm \eps_{\ell})$-sum-approximation $\widetilde f_k^{\ell}$ of $f_k^{\ell}$,
except for the leaf level $\ell = L$, where we compute a $(\delta_L/3)$-sparsified $(1 \pm \delta_L)$-sum-approximation.
The error factors $(1 \pm \eps_{\ell})$ accumulate across the recursion levels $\ell$ (we will see that
$\eps_{\ell}$ depends on $\delta_{\ell}$ and $\eps_{\ell + 1}$ recursively).

To that end, the algorithm first computes
$(\delta_L/3)$-sparsified $(1\pm\delta_L)$-sum-approximations
$\widetilde f_k^L$ for all $1\le k\le 2^L$ using
\Cref{lem:halman-leaf-algorithm}; these correspond to the $2^L=p$
leaves $\{G_1^L,\ldots,G_p^L\}$.
Then, for $\ell = L - 1, \ldots, 0$ and $1 \le k \le 2^{\ell}$, we convolve
$\widetilde f_{2k - 1}^{\ell + 1}$ and $\widetilde f_{2k}^{\ell + 1}$, but only over the
\emph{useful} region of their supports that corresponds to the balanced feasible subsets with
respect to $G_k^{\ell}$. This is the novelty of our work: we save time in the convolution step by
focusing on balanced feasible subsets, while losing only a small fraction of the original feasible
subsets, as guaranteed by \Cref{lem:deterministic-balancing}. We then sparsify the result using
\Cref{lem:sparsify-count-function} to obtain a $\delta_{\ell}$-sparsified
$(1 \pm \eps_{\ell})$-sum-approximation, preventing a blow-up in convolution time at higher levels.

More precisely, our merging step can be formalized by the following lemma, which we prove in the
next subsection.
\newcommand{\epsm}{\varepsilon_{eo}}
\newcommand{\epsc}{\varepsilon_e}
\newcommand{\delm}{\delta_{eo}}
\newcommand{\delc}{\delta_e}
\begin{restatable}[Merging Lemma]{lemma}{merginglemma}
    \label{lem:restricted-merging}
    Let $G\subseteq[n]$ be a group of at least two items, let $f_G$ be the count function of $G$, and let
    $f_{e}$ and $f_{o}$ be the count functions of $G_e$ and $G_o$, respectively. Suppose we are
    given $\delc$-sparsified $(1 \pm \epsc)$-sum-approximations $\widetilde{f_{e}}$ and
    $\widetilde{f_{o}}$ of $f_{e}$ and $f_{o}$ for some $\epsc \in (0, 1/3)$ and $\delc \in (0, 1)$.
    Then, for any $\delm \in (0, 1/8)$, there is an algorithm that computes a $\delm$-sparsified
    $(1 \pm \epsm)$-sum-approximation $\widetilde f_G$ of $f_G$ in time
    \[
        \tcO\left(\log_{|G|} \left(\frac{1}{\delm}\right) \cdot \frac{|G|^{3/2}}{\delc^{2}} +
            \frac{|G|}{\delm}\right)
    \]
    where $\epsm$ satisfies
    \[
        (1 + \epsm) = (1 + \delm)(1 + \epsc)^{2}.
    \]
\end{restatable}
At the first merging level $\ell=L-1$, we apply
\Cref{lem:restricted-merging} with input sparsification parameter
$\delc=\delta_L/3$ and output sparsification parameter
$\delm=\delta_{L-1}$.
At every higher level $\ell<L-1$, we use
$\delc=\delta_{\ell+1}$ and $\delm=\delta_\ell$.

\paragraph{Correctness}
The input requirements of \Cref{lem:restricted-merging} are $\epsc \in (0, 1/3)$, $\delc \in (0, 1)$, and
$\delm \in (0, 1/8)$. To bound the accumulated errors $\eps_{\ell}$, we prove the following fact.
\begin{fact}
    \label{fact:recursion-error-bound}
    A recurrence
    \[
        (1 + \eps_{\ell})
        = (1 + \delta_{\ell})(1 + \eps_{\ell + 1})^{2}
        \qquad \text{for all } 0 \le \ell \le L - 1,
    \]
    with base case
    \[
        (1 + \eps_{L}) = (1 + \delta_{L}).
    \]
    satisfies $0<\eps_\ell<\eps<1/3$ for all
    $0\le\ell\le L$, and
    $\eps_{\ell+1}<\eps_\ell$ for all $0\le\ell<L$.
\end{fact}

\begin{proof}
    First, $0 < \eps_{L}$ is immediate since $\eps_{L} = \delta_{L} > 0$. Next, we show
    $\eps_{0} < \eps$. Solving the recurrence gives
    \[
        (1 + \eps_0) = \prod_{0 \le \ell \le L} (1 + \delta_{\ell})^{2^{\ell}}
        \le \prod_{0 \le \ell \le L} e^{\delta_{\ell}2^{\ell}}
        = e^{\sum_{\ell = 0}^{L}\delta_{\ell}2^{\ell}}
        < e^{\eps/2}
        \le (1 + \eps),
    \]
    where the first inequality uses $x \ge \ln(1 + x)$ for $1 + x > 0$, the last inequality uses
    $x/2 \le \ln(1 + x)$ for all $x \in (0, 1)$ and the strict inequality follows by the fact that,
    for all integers $x \ge 0$,
    \begin{align*}
        \sum_{\ell = 0}^{x}2^{\ell/2}
        < \frac{\sqrt{2}}{\sqrt{2} - 1}2^{x/2}
        = d2^{x/2}.
        \numberthis
        \label{eq:half-power-geometric-series}
    \end{align*}
    so,
    \[
        \sum_{0 \le \ell \le L}2^{\ell}\delta_{\ell} = \sum_{0 \le \ell \le L}
                \frac{\eps}{2dp^{1/2}}2^{\ell/2}
        = \frac{\eps}{2dp^{1/2}} \sum_{0 \le \ell \le L}2^{\ell/2}
        < \frac{\eps}{2dp^{1/2}} \cdot dp^{1/2}
        = \frac{\eps}{2}.
    \]
    Therefore, we have $\eps_0 < \eps$ as desired.

    It remains to show $\eps_{\ell + 1} < \eps_{\ell}$ for all $0 \le \ell < L$. We have
    $\delta_{\ell} > 0$. Hence, from the recurrence,
    \[
        1 + \eps_{\ell}
        = (1 + \delta_{\ell})(1 + \eps_{\ell + 1})^{2}
        > (1 + \eps_{\ell + 1})^{2}
        = 1 + 2\eps_{\ell + 1} + \eps_{\ell + 1}^{2},
    \]
    we have
    \[
        \eps_{\ell + 1}
        < 2\eps_{\ell + 1} + \eps_{\ell + 1}^{2}
        < \eps_{\ell},
    \]
    because $\eps_{\ell + 1} > 0$ inductively. This proves the monotonicity.
\end{proof}

The upper error factors $(1 + \eps_{\ell})$ produced by repeatedly applying \Cref{lem:restricted-merging}
satisfy the recurrence and base case given in \Cref{fact:recursion-error-bound}. Hence,
$\eps_{\ell} \in (0, 1/3)$ for all $\ell$. For the requirement $\delm \in (0, 1/8)$, we note that
$\delta_\ell = \eps/(2dp^{1/2}2^{\ell/2}) < 1/(6d) < 1/8$ since $\eps < 1/3$ and $6d > 8$.
Therefore, we can safely apply \Cref{lem:restricted-merging} at every level of the recursion, up to the root,
with the sparsification parameter $\delta_{\ell}$. Let $\widetilde f_1^0$ be the final
$(1\pm\eps_0)$-sum-approximation of the count
function $f_1^0$ (which is exactly $h$) at the root.
For $L\ge1$, this approximation is also
$\delta_0$-sparsified. Since $\eps_0 < \eps$, we have
\[
    (1 - \eps)|\Omega[t]|
    \le (1 - \eps_0)h^{\le}(t)
    \le \bigl(\widetilde f_1^0\bigr)^{\le}(t)
    \le (1 + \eps_0)h^{\le}(t)
    \le (1 + \eps)|\Omega[t]|,
    \qquad \forall t \ge 0,
\]
and the correctness of our algorithm follows.

\paragraph{Time Complexity}
We now analyze the running time of our algorithm. Each of the $p$ leaves has $\cO(n/p)$ items. By
\Cref{lem:halman-leaf-algorithm}, computing a $\delta_{L}/3$-sparsified $(1 \pm \delta_{L})$-sum-approximation for
each of the $p$ leaves takes
\[
    \cO\left(p \cdot \left(\frac{n}{p}\right)^{3} \bigg/ \delta_{L}\right)
    = \cO\left(\frac{n^{3}}{\eps p}\right),
\]
and each of these sum-approximations has size
\[
    \cO\left(\frac{n}{p} \cdot \frac{1}{\delta_{L}}\right)
    = \cO\left(\frac{n}{\eps}\right).
\]
For recursion levels $0 \le \ell \le L - 1$, fix an internal node $G_k^{\ell}$. By
\Cref{lem:restricted-merging}, the time for one merge at $G_k^{\ell}$ is
\[
    \tcO\Bigl(\log_{|G_k^{\ell}|}\left(\frac{1}{\delta_{\ell}}\right) \cdot
            \frac{|G_k^{\ell}|^{3/2}}{\delta_{\ell + 1}^{2}} +
            \frac{|G_k^{\ell}|}{\delta_{\ell}}\Bigr)
    \le \tcO\Bigl(\frac{|G_k^{\ell}|^{3/2}}{\delta_{\ell + 1}^{2}} +
            \frac{|G_k^{\ell}|}{\delta_{\ell}}\Bigr)
    \le \tcO\Bigl(\frac{|G_k^{\ell}|^{3/2}}{\delta_{\ell + 1}^{2}}\Bigr)
    \le \tcO\Bigl(\frac{n^{3/2}}{2^{3\ell/2} \delta_{\ell + 1}^{2}}\Bigr),
\]
where the inequalities follow by $\Omega (\eps/n) \le \delta_{\ell} < 1$ and
$\delta_{\ell} = \sqrt{2}\delta_{\ell + 1} \ge \delta_{\ell + 1}^2$ and
$|G_k^{\ell}| = \Theta(n/2^{\ell})$. There are $2^{\ell}$ internal nodes at level $\ell$, so summing
over all levels gives
\[
    \sum_{\ell = 0}^{L - 1} 2^{\ell} \cdot \tcO\Bigl(\frac{n^{3/2}}{2^{3\ell/2} \delta_{\ell +
            1}^{2}}\Bigr)
    = \tcO\left(\frac{n^{3/2}p}{\eps^{2}} \sum_{\ell = 0}^{L} 2^{\ell/2}\right)
    = \tcO\left(\frac{n^{3/2}p^{3/2}}{\eps^{2}}\right),
\]
where the last equality uses \Cref{eq:half-power-geometric-series}. Choose $p$ to be the largest power of two at most
$n^{3/5}\eps^{2/5}$. Since $n^{-3/2} \le \eps < 1$, we have $1 \le p \le n$ as required and
$p = \Theta(n^{3/5}\eps^{2/5})$. Thus, the total running time is
\[
    \tcO\left(
        \frac{n^3}{\eps p}
        + \frac{n^{3/2}p^{3/2}}{\eps^2}
    \right)
    = \tcO\left(\frac{n^{2.4}}{\eps^{1.4}}\right).
\]
This proves the following theorem.
\fptas*

\subsection{Merging Lemma}
In this subsection, we prove the Merging Lemma (\Cref{lem:restricted-merging}). To do so, we define
\emph{useful pairs}. These are pairs in the supports of the two input sum-approximations that
represent the contributions from each side of a \emph{balanced} subset at the parent node, and we
convolve only over such pairs. More specifically, fix $G \subseteq [n]$ with $|G| \ge 2$, and let $f_G$ be the count
function of $G$, and let $f_{e}$ and $f_{o}$ be the count functions of $G_e$ and $G_o$,
respectively. Suppose we have already computed $\delc$-sparsified $(1 \pm \epsc)$-sum-approximations
$\widetilde f_{e}$ and $\widetilde f_{o}$ for some $\epsc \in (0, 1/3)$ and $\delc \in (0, 1)$. Fix
a sparsification parameter $\delm \in (0, 1/8)$.

\newcommand{\cg}{%
    1 + \log_{|G|}\left(
        \frac{96}{(1 - \epsc)^2 \bigl(1 - \gamma/(1 - \epsc)^2\bigr)}
    \right)%
}
\newcommand{\CG}{\left\lceil 32c\sqrt{|G|\log |G|} + 2\right\rceil}
\newcommand{\s}{2 - (1 + \delm)(1 + \epsc)^2}

Recall \Cref{def:balanced-subsets}. We define the \emph{balancing} parameters by
\[
    c := \cg \qquad \text{where} \qquad \gamma := \s,
\]
and
\[
    C := \CG.
\]
We first show that $c$ is well-defined and derive bounds on $c$.

\begin{fact}
    \label{fact:balancing-parameter-range}
    The balancing parameter $c$ is well-defined and
    \[
        1 < c \le \cO\left(1 + \log_{|G|} \left(\frac{1}{\delm}\right)\right).
    \]
\end{fact}

\begin{proof}
    By expanding $\gamma = 2 - (1 + \delm)(1 + \epsc)^2$, we have
    \[
        1 - \frac{\gamma}{(1 - \epsc)^2}
        = \frac{(1 - \epsc)^2 - \gamma}{(1 - \epsc)^2}
        = \frac{(1 - \epsc)^2 - \bigl(2 - (1 + \delm)(1 + \epsc)^2\bigr)}{(1 - \epsc)^2}
        = \frac{\delm(1 + \epsc)^2 + 2\epsc^2}{(1 - \epsc)^2}.
    \]
    Hence,
    \[
        0 < (1 - \epsc)^2 \left(1 - \frac{\gamma}{(1 - \epsc)^2}\right) = \bigl(\delm(1 + \epsc)^2 +
            2\epsc^2\bigr) < 96
    \]
    where the inequalities follow from $\epsc \in (0, 1)$ and $\delm \in (0, 1)$. This implies that
    $c$ is well-defined and that $c > 1$. For the upper bound, we have
    \[
        \frac{96}{(1 - \epsc)^2 \bigl(1 - \gamma/(1 - \epsc)^2\bigr)} =
        \frac{96}{\bigl(\delm(1 + \epsc)^2 + 2\epsc^2\bigr)}
        \le \cO\left(\frac{1}{\delm}\right),
    \]
    since $\epsc > 0$. This upper bound implies $c \le \cO\bigl(1 + \log_{|G|}(1/\delm)\bigr)$.
\end{proof}

Next, we formally define \emph{useful pairs}, together with the convolution restricted to them.
\begin{definition}[Useful pairs]
    \label{def:useful-pairs}
    We say that $(x, y) \in \supp(\widetilde f_{e}) \times \supp(\widetilde f_{o})$ is a
    \emph{useful pair} if
    \begin{itemize}
        \item $\widetilde f_{e}^{\le}(x) \le 6|G|^{C} \widetilde f_{e}^{\le}(y)$; and
        \item $\widetilde f_{o}^{\le}(y) \le 6|G|^{C} \widetilde f_{o}^{\le}(x)$.
    \end{itemize}
    Let $U \subseteq \supp(\widetilde f_{e}) \times \supp(\widetilde f_{o})$ be the set of all
    useful pairs.
\end{definition}
In \Cref{def:useful-pairs}, we stress that we counter-intuitively use $\widetilde f_{e}^{\le}(y)$
and $\widetilde f_{o}^{\le}(x)$ instead of $\widetilde f_{o}^{\le}(y)$ and
$\widetilde f_{e}^{\le}(x)$ on the right-hand side of the inequalities, even though
$x \in \supp(\widetilde f_{e})$ and $y \in \supp(\widetilde f_{o})$.

\begin{definition}[Convolution on useful pairs]
    \label{def:convolution-on-useful-pairs}
    Define the convolution of $\widetilde f_e$ and $\widetilde f_o$ restricted to useful pairs by
    \[
        \widetilde f(t) := \sum_{\substack{x + y = t \\ (x, y) \in U}} \widetilde f_e(x)\widetilde
            f_o(y).
    \]
\end{definition}

With our choice of $c$, the number of useful pairs is much smaller than the total number of pairs in
the supports, smaller by roughly a factor of $|G|^{1/2}$ compared to the standard convolution.
Consequently, restricting the convolution to useful pairs improves the time complexity of computing
a sum-approximation of $f_G$.

\begin{lemma}
    \label{lem:compute-useful-pairs}
    The restricted convolution $\widetilde f$ in \Cref{def:convolution-on-useful-pairs} can be computed in
    time and space
    \[
        \tcO\left(\log_{|G|} \left(\frac{1}{\delm}\right) \cdot \frac{|G|^{3/2}}{\delc^{2}}\right).
    \]
\end{lemma}

\begin{proof}
    Once $U$ is available, we can compute $\widetilde f$ by iterating over all $(x, y) \in U$ and
    adding $\widetilde f_{e}(x)\widetilde f_{o}(y)$ to the entry for $t = x + y$; this takes
    $\tcO(|U|)$ time and $\cO(|U|)$ space. Hence, it suffices to show that the time and space
    complexities of constructing the set $U$ are at most the stated time bound.

    We store the supports in sorted order, together with their
    prefix sums. For each fixed support element, we use binary
    search to locate the first candidate in the opposite support
    and scan forward until the prefix-sum condition fails. Consider the following procedure:
    \begin{enumerate}
        \item Initialize $S \gets \varnothing$.
        \item For each $x \in \supp(\widetilde f_{e})$, scan and include all
            $y \in \supp(\widetilde f_{o})$ such that $y \ge x$ and
            $\widetilde f_{o}^{\le}(y) \le 6|G|^{C} \widetilde f_{o}^{\le}(x)$.
        \item For each $y \in \supp(\widetilde f_{o})$, scan and include all
            $x \in \supp(\widetilde f_{e})$ such that $x \ge y$ and
            $\widetilde f_{e}^{\le}(x) \le 6|G|^{C} \widetilde f_{e}^{\le}(y)$
    \end{enumerate}

    We claim that $S$ is exactly the set of useful pairs in $U$. Fix $(x, y) \in S$. If $y \ge x$,
    the algorithm ensures the second condition of useful pairs:
    $\widetilde f_{o}^{\le}(y) \le 6|G|^{C} \widetilde f_{o}^{\le}(x)$ (\Cref{def:useful-pairs}).
    Moreover, since $\widetilde f_{e}^{\le}$ is non-decreasing, we have
    \[
        \widetilde f_{e}^{\le}(x) \le \widetilde f_{e}^{\le}(y) \le 6|G|^{C} \widetilde
            f_{e}^{\le}(y),
    \]
    so the first condition holds as well. Hence, $(x, y) \in U$. The case $x \ge y$ is similar, as
    the algorithm ensures that $\widetilde f_{e}^{\le}(x) \le 6|G|^{C} \widetilde f_{e}^{\le}(y)$,
    and monotonicity of $\widetilde f_{o}^{\le}$ implies the second condition of useful pairs.
    Conversely, fix $(x, y) \in U$. If $x \ge y$, the algorithm includes $(x, y)$ in $S$ in the
    third step, by the first condition of useful pairs. If $y \ge x$, the algorithm includes
    $(x, y)$ in $S$ in the second step, by the second condition of useful pairs.

    Next, we analyze the time and space complexity of constructing $U$. Let
    $\supp(\widetilde f_{e}) = \{b_1 < \cdots < b_{m_e} \}$ and
    $\supp(\widetilde f_{o}) = \{c_1 < \cdots < c_{m_o} \}$. Since $\widetilde f_{e}$ and
    $\widetilde f_{o}$ are $\delc$-sparsified, \Cref{lem:sparsify-count-function} implies that
    \begin{equation}
        \label{eq:even-side-support-size}
        m_e = \tcO\left(\frac{|G_e|}{\delc}\right) = \tcO\left(\frac{|G|}{\delc}\right)
    \end{equation}
    and
    \begin{equation}
        \label{eq:odd-side-geometric-growth}
        \widetilde f_o^{\le}(c_i) \ge (1 + \delc)\widetilde f_o^{\le}(c_{i - 1}) \quad \text{for all
            } 2 \le i \le m_o.
    \end{equation}
    Fix $x \in \supp(\widetilde f_e)$. Note that the number of $c_i \in \supp(\widetilde f_o)$ such
    that $c_i \ge x$ and $\widetilde f_o^{\le}(c_i) \le 6|G|^{C} \widetilde f_o^{\le}(x)$ is at most
    $\cO\bigl(\log_{1 + \delc}(6|G|^C)\bigr)$ by \Cref{eq:odd-side-geometric-growth}, and this is at most
    \[
        \tcO\bigl(\log_{(1 + \delc)}(6|G|^{C})\bigr)
        \le \tcO\left(\frac{C}{\delc}\right)
        \le \tcO\left(\frac{c\sqrt{|G|}}{\delc}\right)
        \le \tcO\left(\log_{|G|}\left(\frac{1}{\delm}\right) \cdot
                \frac{\sqrt{|G|}}{\delc}\right)
    \]
    by \Cref{fact:balancing-parameter-range} and $C = \CG$. Since $|\supp(\widetilde f_{e})| \le \tcO(|G|/\delc)$ by
    \Cref{eq:even-side-support-size}, the time and space complexities of step $2$ are at most
    \[
        \tcO\left(\frac{|G|}{\delc}\right) \cdot
                \tcO\left(\log_{|G|}\left(\frac{1}{\delm}\right) \cdot
                \frac{\sqrt{|G|}}{\delc}\right)
        = \tcO\left(\log_{|G|}\left(\frac{1}{\delm}\right) \cdot
                \frac{|G|^{3/2}}{\delc^{2}}\right).
    \]
    The time and space complexities of step $3$ can be bounded in the same way by a similar
    argument, and the proof follows.
\end{proof}

We now analyze how much $\widetilde f$ \emph{underestimates} the number of feasible subsets of $G$.
To this end, we bound the difference between the full-range convolution and the restricted
convolution,
\[
    \Delta(t) := \bigl(\widetilde f_e \star \widetilde f_o\bigr)^{\le}(t) - \widetilde f^{\le}(t),
\]
in terms of the prefix sum of the counting function $f_G$, which counts the number of feasible
subsets with respect to $G$ and $t$. The proof uses \Cref{lem:deterministic-balancing} to bound the contributions
of unbalanced subsets.
\begin{lemma}
    \label{lem:omitted-convolution-mass}
    \[
        0 \le \Delta(t) \le \frac{96}{|G|^{c - 1}} \cdot f_G^{\le}(t).
    \]
\end{lemma}

\begin{proof}
    Fix an arbitrary $t \ge 0$. For the lower bound, we observe that
    \[
        \begin{aligned}
            \Delta(t) &:= \bigl(\widetilde f_e \star \widetilde f_o\bigr)^{\le}(t) - \widetilde
                f^{\le}(t) \\
            &= \sum_{x + y \le t}\widetilde f_e(x)\widetilde f_o(y) - \sum_{\substack{x + y \le t \\
                (x, y) \in U}}\widetilde f_e(x)\widetilde f_o(y) \\
            &= \sum_{\substack{x + y \le t \\ (x, y) \notin U}} \widetilde f_e(x)\widetilde f_o(y)
                \\
            &\ge 0.
        \end{aligned}
    \]
    For the upper bound, we expand $\Delta(t)$ using the definition of useful pairs
    (\Cref{def:useful-pairs}):
    \begin{align*}
        \Delta(t)
        = \sum_{\substack{x + y \le t \\ (x, y) \notin U}} \widetilde f_e(x)\widetilde f_o(y)
        =
        \underbrace{\sum_{\substack{x + y \le t \\ \widetilde f_e^{\le}(x) > 6|G|^{C}\widetilde
            f_e^{\le}(y)}} \widetilde f_e(x)\widetilde f_o(y)}_{S_e}
        +
        \underbrace{
        \sum_{\substack{x + y \le t \\ \widetilde f_o^{\le}(y) > 6|G|^{C}\widetilde f_o^{\le}(x)}}
            \widetilde f_e(x)\widetilde f_o(y)}_{S_o}.
        \numberthis
        \label{eq:excluded-pair-decomposition}
    \end{align*}
    Indeed, the two failure conditions are disjoint: the first requires $x > y$, while the second
    requires $y > x$.

    To upper-bound $\Delta(t)$, we introduce some definitions. For any $z \ge 0$, let $\Omega_e[z]$
    denote the collection of feasible subsets with respect to $G_e$ and $z$, and let $\Omega_o[z]$
    denote the collection of feasible subsets with respect to $G_o$ and $z$. Define a map
    $\pi_{e} \colon \Omega_e[t] \to \Omega_e[t]$ that removes the heaviest $C$ items from
    $I \in \Omega_e[t]$ (if $I$ has fewer than $C$ items, remove all of them). Similarly, define a
    map $\pi_o \colon \Omega_o[t] \to \Omega_o[t]$ that does the same. These maps are well-defined,
    since removing items does not increase the weight. Lastly, we define
    \[
        D_e(t, y) := \bigl\{S \in \Omega_{e}[t] : y \le w(S) \le t - y \text{ and } \pi_e(S) \notin
            \Omega_{e}[y]\bigr\}
    \]
    and
    \[
        D_o(t, x) := \bigl\{S \in \Omega_{o}[t] : x \le w(S) \le t - x \text{ and } \pi_o(S) \notin
            \Omega_{o}[x]\bigr\}.
    \]
    Intuitively, $D_e(t, y)$ is the collection of subsets $S \in \Omega_{e}[t - y]$ such that, even
    after removing the heaviest $C$ items from $S$, the resulting set still has weight larger than
    $y$; equivalently, $S$ cannot form a balanced feasible subset of $G$ together with a subset in
    $G_{o}$ of weight $y$. Hence, by \Cref{def:count-function} and \Cref{def:balanced-subsets}, we have

    \begin{align*}
        |\Omega_G[t] \setminus \Gamma_G[t]|
        = \left(\sum_{y \le t} f_{o}(y) \cdot |D_e(t, y)|\right)
        + \left(\sum_{x \le t} f_{e}(x) \cdot |D_o(t, x)|\right).
        \numberthis
        \label{eq:unbalanced-subset-counts}
    \end{align*}

    Now, we bound $S_e$ in \Cref{eq:excluded-pair-decomposition} in terms of $D_e(t, y)$ and $f_o$.
    \begin{claim}
        \label{claim:excluded-pair-bound}
        \[
            S_e \le 3(1 + \epsc)^2\sum_{y \le t} f_o(y) \cdot |D_e(t, y)|.
        \]
        \[
            S_o \le 3(1 + \epsc)^2\sum_{x \le t} f_e(x) \cdot |D_o(t, x)|.
        \]
    \end{claim}
    \begin{proof}
        We only prove the bound for $S_e$, as the bound for $S_o$ follows analogously. Rewrite $S_e$
        as
        \[
                S_e := \sum_{\substack{x + y \le t \\ \widetilde f_e^{\le}(x) > 6|G|^{C}\widetilde
                    f_e^{\le}(y)}} \widetilde f_e(x)\widetilde f_o(y)
                = \sum_{x \le t} \widetilde f_e(x) \sum_{\substack{y \le t - x \\ \widetilde
                    f_e^{\le}(x) > 6|G|^{C}\widetilde f_e^{\le}(y)}} \widetilde f_o(y)
        \]
        For each fixed $x$, the allowed values of $y$ form a prefix interval: both $y \le t - x$ and
        the upper bound on $\widetilde f_e^{\le}(y)$ are preserved when $y$ decreases. Therefore,
        the upper prefix bound for the $(1 \pm \epsc)$-sum-approximation $\widetilde f_o$ gives

        \begin{align*}
            S_e = \sum_{x \le t} \widetilde f_e(x) \sum_{\substack{y \le t - x \\ \widetilde
                f_e^{\le}(x) > 6|G|^{C}\widetilde f_e^{\le}(y)}} \widetilde f_o(y) & \le (1 + \epsc)
                \sum_{x \le t} \widetilde f_e(x) \sum_{\substack{y \le t - x \\ \widetilde
                f_e^{\le}(x) > 6|G|^{C}\widetilde f_e^{\le}(y)}} f_o(y) \\
            & = (1 + \epsc)\sum_{\substack{x + y \le t \\ \widetilde f_e^{\le}(x) >
                6|G|^{C}\widetilde f_e^{\le}(y)}} f_o(y) \cdot \widetilde f_e(x)
        \end{align*}
        The summation condition $\widetilde f_e^{\le}(x) > 6|G|^{C} \widetilde f_e^{\le}(y)$ implies
        \[
            \frac{4}{3}f_e^{\le}(x)
            > (1 + \epsc) f_e^{\le}(x)
            \ge \widetilde f_e^{\le}(x)
            > 6|G|^{C} \widetilde f_e^{\le}(y)
            \ge (1 - \epsc)6|G|^{C} f_e^{\le}(y)
            \ge \frac{12}{3}|G|^{C} f_e^{\le}(y)
        \]
        by \Cref{def:sum-approximation} and the fact that we assumed $\epsc \le 1/3$. Hence, we have
        \begin{align*}
            S_e & \le (1 + \epsc)\sum_{\substack{x + y \le t \\ \widetilde f_e^{\le}(x) >
                6|G|^{C}\widetilde f_e^{\le}(y)}} f_o(y) \cdot \widetilde f_e(x) \le (1 +
                \epsc)\sum_{\substack{x + y \le t \\ f_e^{\le}(x) > 3|G|^{C} f_e^{\le}(y)}} f_o(y)
                \cdot \widetilde f_e(x)
        \end{align*}
        We further bound $S_e$ in terms of $f_e^{\le}$ as follows:
        \begin{align*}
            S_e \le (1 + \epsc)\sum_{\substack{x + y \le t \\ f_e^{\le}(x) > 3|G|^{C} f_e^{\le}(y)}}
                f_o(y) \cdot \widetilde f_e(x)
            & \le (1 + \epsc)\sum_{\substack{y \le t \\ f_e^{\le}(t - y) > 3|G|^C
                f_e^{\le}(y)}}f_o(y) \sum_{x \le t - y} \widetilde f_e(x) \\
            & \le (1 + \epsc)^2\sum_{\substack{y \le t \\ f_e^{\le}(t - y) > 3|G|^C f_e^{\le}(y)}}
                f_o(y) \cdot f_e^{\le}(t - y).
        \end{align*}
        The summation condition $f_e^{\le}(t - y) > 3|G|^C f_e^{\le}(y)$ implies
        \[
            \frac{1}{3}f_e^{\le}(t - y) > |G|^C f_e^{\le}(y) \ge f_e^{\le}(y).
        \]
        Hence, if we define
        \[
            \alpha(t, y) := \max\bigl\{0, f_e^{\le}(t - y) - (1 + |G|^{C})f_e^{\le}(y)\bigr\},
        \]
        then we have
        \[
            f_e^{\le}(t - y)
            <
            \max\Bigl\{0, 3\bigl(f_e^{\le}(t - y) - f_e^{\le}(y) - |G|^{C} f_e^{\le}(y)\bigr)\Bigr\}
            = 3\alpha(t, y).
        \]
        This yields the further bound
        \begin{align*}
            S_e
            \le (1 + \epsc)^2\sum_{\substack{y \le t \\ f_e^{\le}(t - y) > 3|G|^Cf_e^{\le}(y)}}
                f_o(y) \cdot f_e^{\le}(t - y)
            & \le 3(1 + \epsc)^2\sum_{\substack{y \le t \\ f_e^{\le}(t - y) > 3|G|^C f_e^{\le}(y)}}
                f_o(y) \cdot \alpha(t, y) \\
            & \le 3(1 + \epsc)^2\sum_{y \le t} f_o(y) \cdot \alpha(t, y).
        \end{align*}
        Now, recall
        \[
            D_e(t, y) := \bigl\{S \in \Omega_{e}[t] : y \le w(S) \le t - y \text{ and } \pi_e(S)
                \notin \Omega_{e}[y]\bigr\}.
        \]
        We note that every $I\in\Omega_e[y]$ has at most
        \[
            \sum_{j=0}^{\min\{C,|G_e|\}}
            \binom{|G_e|}{j}
            \le (|G_e|+1)^C
            \le |G|^C
        \]
        preimages under $\pi_e$, corresponding to the possible
        choices of at most $C$ removed items.
        Consequently,
        \[
            |(\pi_e)^{-1}(\Omega_e[y])|
            \le |G|^C f_e^{\le}(y)
        \]
        which implies that
        \begin{align*}
            |D_e(t, y)|
            & \ge \max\bigl\{0, |\Omega_{e}[t - y]| - |\Omega_{e}[y - 1]| -
                |(\pi_e)^{-1}(\Omega_{e}[y])|\bigr\} \\
            & \ge \max\bigl\{0, f_{e}^{\le}(t - y) - f_{e}^{\le}(y - 1) - |G|^{C}
                f_{e}^{\le}(y)\bigr\} \\
            & \ge \max\bigl\{0, f_{e}^{\le}(t - y) - f_{e}^{\le}(y) - |G|^{C} f_{e}^{\le}(y)\bigr\}
                \\
            & = \alpha(t, y).
        \end{align*}
        Therefore, we obtain the desired bound:
        \[
            S_e \le 3(1 + \epsc)^2\sum_{y \le t} f_o(y) \cdot \alpha(t, y) \le 3(1 + \epsc)^2\sum_{y
                \le t} f_o(y) \cdot |D_e(t, y)|.
        \]
    \end{proof}

    Combining \Cref{eq:excluded-pair-decomposition}, \Cref{claim:excluded-pair-bound}, and \Cref{eq:unbalanced-subset-counts}, and using the fact that
    $\epsc < 1$, we obtain an upper bound on $\Delta(t)$:
    \begin{align*}
        \Delta(t)
        = S_e + S_o
        \le 3(1 + \epsc)^2
        \left(
        \sum_{y \le t} f_o(y)|D_e(t, y)|
        +
        \sum_{x \le t} f_e(x)|D_o(t, x)|
        \right)
        & = 3(1 + \epsc)^2 |\Omega_G[t] \setminus \Gamma_G[t]| \\
        & \le 12 \cdot |\Omega_G[t] \setminus \Gamma_G[t]|
        \numberthis
        \label{eq:restricted-convolution-error-sum}
    \end{align*}
    By \Cref{fact:balancing-parameter-range}, we have $c > 1$, so \Cref{lem:deterministic-balancing} applies with our choice of $C$
    and gives
    \[
        |\Omega_{G}[t] \setminus \Gamma_G[t]| \le \frac{8}{|G|^{c - 1}} |\Omega_{G}[t]| \le
            \frac{8}{|G|^{c - 1}} f_G^{\le}(t) .
    \]
    Therefore, combining this with \Cref{eq:restricted-convolution-error-sum}, we conclude the desired inequality:
    \begin{align*}
        \Delta(t)
        & \le 12 \cdot |\Omega_G[t] \setminus \Gamma_G[t]| \le \frac{96}{|G|^{c - 1}} \cdot
            f_G^{\le}(t).
    \end{align*}
\end{proof}

It remains to combine our results and prove the Merging Lemma (\Cref{lem:restricted-merging}), which formalizes
the merging step of our algorithm for computing a $\delta_{\ell}$-sparsified
$(1 + \eps_{\ell})$-sum-approximation to the count function $f_{k}^{\ell}$.
\merginglemma*

\begin{proof}
    Let $\widetilde f$ be the convolution of $\widetilde f_e$ and $\widetilde f_o$ restricted to
    $U$, as in \Cref{def:convolution-on-useful-pairs}. By \Cref{lem:compute-useful-pairs}, the space and
    time complexities of computing $\widetilde f$ are at most
    \[
        \tcO\left(\log_{|G|}\left(\frac{1}{\delm}\right) \cdot \frac{|G|^{3/2}}{\delc^{2}}\right).
    \]
    We then apply $\delm$-sparsification to $\widetilde f$. By \Cref{lem:sparsify-count-function},
    $\delm$-sparsification takes
    \[
        \tcO\left(\log_{|G|}\left(\frac{1}{\delm}\right) \cdot \frac{|G|^{3/2}}{\delc^{2}} +
            \frac{|G|}{\delm}\right)
    \]
    time, and we obtain the desired time complexity.

    Observe that $\delm$-sparsification introduces a multiplicative approximation factor of
    $(1 + \delm)$. Recall $\gamma := \s$. Clearly, by the input premises, we have
    \[
        \gamma = \s > 2 - (9/8)(16/9) = 0
    \]
    To prove correctness, it suffices to show that, \emph{before} $\delm$-sparsification,
    $\widetilde f$ is a $((1 + \epsc)^2, \gamma)$-sum-approximation of $f_G$. Then, after
    $\delm$-sparsification, \Cref{lem:sparsify-count-function} implies that its prefix sum satisfies
    \[
        (\s) f_G^{\le}(t) = \gamma f_G^{\le}(t)
        \le
        \widetilde f^{\le}(t)
        \le
        (1 + \delm)(1 + \epsc)^2f_G^{\le}(t),
    \]
    and hence it satisfies
    \[
        (1 + \epsm) = (1 + \delm)(1 + \epsc)^2
        \quad \text{and} \quad
        (1 - \epsm) = \s = \gamma.
    \]
    To show that $\widetilde f$ is a $((1 + \epsc)^2, \gamma)$-sum-approximation of $f_G$, we note
    that the \emph{full-range} convolution of the two sum-approximations $\widetilde f_e$ and
    $\widetilde f_o$ satisfies
    \begin{equation}
        \label{eq:full-convolution-error-bounds}
        (1 - \epsc)^2 f_G^{\le}(t)
        \le
        \bigl(\widetilde f_e \star \widetilde f_o\bigr)^{\le}(t)
        \le
        (1 + \epsc)^2 f_G^{\le}(t)
    \end{equation}
    by \Cref{lem:full-convolution-approximation}. On the other hand, we have
    \begin{equation}
        \label{eq:omitted-convolution-error-bound}
        0 \le \bigl(\widetilde f_e \star \widetilde f_o\bigr)^{\le}(t) - \widetilde f^{\le}(t)
        \le
        \frac{96}{|G|^{c - 1}} f_G^{\le}(t)
    \end{equation}
    by \Cref{lem:omitted-convolution-mass}. Combining \eqref{eq:full-convolution-error-bounds} and \eqref{eq:omitted-convolution-error-bound}, we
    obtain the desired lower bound on $\widetilde f$:
    \begin{align*}
        \widetilde f^{\le}(t)
        & \ge \bigl(\widetilde f_e \star \widetilde f_o\bigr)^{\le}(t)
        - \frac{96}{|G|^{c - 1}} f_G^{\le}(t)
        \ge (1 - \epsc)^2 f_G^{\le}(t)
        - \frac{96}{|G|^{c - 1}} f_G^{\le}(t)
        = \gamma f_G^{\le}(t)
    \end{align*}
    since
    \[
        c = \cg
    \]
    For the upper bound, we have
    \[
        \widetilde f^{\le}(t)
        \le
        \bigl(\widetilde f_e \star \widetilde f_o\bigr)^{\le}(t)
        \le
        (1 + \epsc)^2 f_G^{\le}(t),
    \]
    again by \Cref{lem:full-convolution-approximation} and \Cref{eq:full-convolution-error-bounds}. Therefore, we conclude that
    $\widetilde f$ is a $((1 + \epsc)^2, \gamma)$-sum-approximation of $f_G$ before
    $\delm$-sparsification, as claimed.
\end{proof}

\newpage
\bibliographystyle{alpha}
\bibliography{ref}
\end{document}